\documentclass[journal]{IEEEtran}

\usepackage{graphicx} 
\usepackage{graphics} 
\usepackage{epsfig}   
\usepackage{amsmath}  
\usepackage{amssymb}  
\usepackage{amsthm}  
\usepackage{subfigure}
\usepackage{psfrag}
\usepackage{url}
\usepackage[usenames,dvipsnames]{color}
\usepackage{gastex}
\usepackage{array}
\usepackage{xcolor}

\newtheorem{lemma}{Lemma}
\newtheorem{definition}{Definition}
\newtheorem{remark}{Remark}
\newtheorem{theorem}{Theorem}
\newtheorem{corollary}{Corollary}

\newtheorem{algorithm}{Algorithm}
\newtheorem{assumption}{Assumption}
\newcommand{\fragsize}{0.8}

\newcommand{\rate}{p}
\newcommand{\optrate}{p^*}
\newcommand{\maxI}{\gamma}
\newcommand{\onevec}{\mathbf{1}}
\newcommand{\nomI}{\mathcal{I}^0}
\newcommand{\eg}{\emph{e.g.}}
\newcommand{\ie}{\emph{i.e.}}
\newcommand{\bi}{{(i)}}
\newcommand{\bim}{{(i-1)}}
\newcommand{\bie}{{(1)}}
\newcommand{\bir}{{(\rate)}}
\newcommand{\biro}{{(\optrate)}}
\newcommand{\birl}{{(\rate_\ell)}}
\newcommand{\loopset}{\mathbb{L}}
\newcommand{\edit}[1]{{#1}}
\newcommand{\editt}[1]{{#1}}

\begin{document}

\title{Multiple Loop Self-Triggered Model Predictive Control for Network Scheduling and Control}

\author{Erik Henriksson, Daniel E. Quevedo, Edwin G.W. Peters, Henrik Sandberg, Karl Henrik Johansson
\thanks{E. Henriksson, H. Sandberg and K.H. Johansson are with ACCESS Linnaeus Centre, School of Electrical Engineering, Royal Institute of Technology, 10044 Stockholm,
        Sweden. e-mail:~\{erike02,~hsan,~kallej\}@ee.kth.se}
\thanks{D.E. Quevedo and E.G.W. Peters are with The University of Newcastle, NSW 2308, Australia. e-mail:~dquevedo@ieee.org,~edwin.g.w.peters@gmail.com.}
\thanks{This work was supported by the Swedish Governmental Agency
        for Innovation Systems through the WiComPi project, the Swedish Research Council under
        Grants 2007-6350 and 2009-4565, the Knut and Alice Wallenberg
        Foundation, and the Australian Research Council’s Discovery Projects funding scheme (project number DP0988601).}}

\maketitle

\begin{abstract}
We present an algorithm for controlling \edit{and scheduling} multiple linear time-invariant
processes on a shared \edit{bandwidth limited} communication network using adaptive sampling intervals. 
\edit{The controller is centralized and computes at every sampling instant not only the new control command for a process, but also decides the time interval to wait until taking the next sample.} 
\edit{The
approach relies on model predictive control ideas, where the cost function penalizes the state and control effort as well as the time interval until the next sample is taken.}
 The latter is
introduced in order to generate an adaptive sampling scheme for the overall
system such that the sampling time increases as the \editt{norm of the system state goes to zero}. The paper presents a method for synthesizing such a predictive
controller and gives explicit sufficient conditions for when it is
stabilizing. Further explicit conditions are given which guarantee conflict
free transmissions on the network. It is shown that the optimization problem
may be solved off-line and that the controller can be implemented as a lookup
table of state feedback gains. \edit{Simulation studies which compare the proposed algorithm to periodic sampling illustrate potential performance gains.}
\end{abstract}


\section{Introduction}
Wireless sensing and control systems have received increased attention in the
process industry over the last years. Emerging technologies in low-power
wake-up radio enables engineering of a new type of industrial automation
systems where sensors, controllers and actuators communicate over a wireless
channel. The introduction of a wireless medium in the control loop gives rise
to new challenges which need to be handled \cite{will08}. The aim of this paper
is to address the problem of how the medium access to the wireless channel
could be divided between the loops, taking the process dynamics into
consideration. 
\edit{We investigate possibilities to design a self-triggered controller, that adaptively chooses the sampling period for multiple control loops. The aim is to reduce the amount of generated network traffic, while maintaining a guaranteed level of performance in respect of driving the initial system states to zero and the control effort needed.
}

Consider the networked control system in Fig.~\ref{fig:NCS}, which shows how
the sensors and the controller are connected through a wireless network. The
wireless network is controlled by a \emph{Network Manager} which allocates
medium access to the sensors and triggers their transmissions. This setup is
motivated by current industry standards \edit{based on the IEEE 802.15.4 standard}, \eg, \cite{whart,ieee:802.15.4-2011:standard}, which utilizes this
structure for wireless control in process industry. Here the triggering is in
turn generated by the controller which, in addition to computing the
appropriate control action, dynamically determines the time of the next sample
by a self-triggering approach \cite{vel+03}. In doing so the controller gives
varying attention to the loops depending on their state, while trying to
communicate only few samples. To achieve this the controller must, for every
loop, trade control performance against inter sampling time and give a
quantitative measure of the resulting performance.
\begin{figure}[tb]
    \centering
    \psfrag{C}[][][1]{$\mathcal{C}$}
    \psfrag{P}[][][1]{$\mathcal{P}_1$}
    \psfrag{S}[][][1]{$\mathcal{S}_1$}
    \psfrag{A}[][][1]{$\mathcal{A}_1$}
    \psfrag{Q}[][][1]{$\mathcal{P}_s$}
    \psfrag{T}[][][1]{$\mathcal{S}_s$}
    \psfrag{B}[][][1]{$\mathcal{A}_s$}
    \includegraphics[width=0.85\columnwidth]{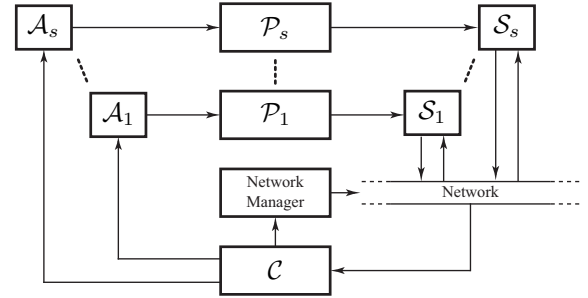}\\
    \caption{Actuators $\mathcal{A}$ and processes $\mathcal{P}$ are wired to the controller $\mathcal{C}$
    while the sensors $\mathcal{S}$ communicate over a wireless network, which in turn is coordinated by the \emph{Network Manager}.}\label{fig:NCS}
\end{figure}

The main contribution of the paper is to show that a self-triggering controller
can be derived using a receding horizon control formulation where the predicted
cost is used to jointly determine what control signal to be applied as well as
the time of the next sampling instant. Using this formulation we can guarantee
a minimum and a maximum time between samples.

We will initially consider a single-loop system. We will then extend the
approach to the multiple-loop case which can be analyzed with additional
constraints on the communication pattern. The results presented herein are
extensions to the authors' previous work presented in \cite{hen+12}. These
results have been extended to handle multiple control loops on the same network
while maintaining control performance and simultaneously guarantee conflict
free transmissions on the network.

The development of control strategies for wireless automation has become a
large area of research in which, up until recently, most efforts have been made
under the assumption of periodic communication \cite{ant+bai07}. However the
idea of adaptive sampling is receiving increased attention. The efforts within
this area may coarsely be divided into the two paradigms of event- and
self-triggered control. In event-triggered control, \eg, \cite{ast+ber99},
\cite{arz99}, \cite{tab07}, \cite{hee+08}, \cite{bemporad:2006:event-based-mpc}, \cite{varutti:2009:event-based-mpc-ncs}, \cite{gross:2013:distributed-mpc-event-based}, \cite{molin:2014:bi-level-event-trigger-shared}, \cite{blind:2013:time-trigger-and-event-based-integrator} the sensor continuously monitors the
process state and generates a sample when the state violates some predefined
condition. Self-triggered control, \eg, \cite{vel+03}, \cite{wan+lem09},
\cite{ant+tab10}, \cite{maz+10}, utilizes a model of the system to predict when
a new sample needs to be taken in order to fulfill some pre-defined condition.
A possible advantage of event- over self-triggered control is that the
continuous monitoring of the state guarantees that a sample will be drawn as
soon as the design condition is violated, thus resulting in an appropriate
control action. The self-triggered controller will instead operate in open-loop
between samples. This could potentially be a problem as disturbances to the
process between samples cannot be attenuated. This problem may however be
avoided by good choices of the inter-sampling times. The possible advantage of
self- over event-triggered control is that the transmission time of sensor
packets is known a-priori and hence we may schedule them, enabling sensors and
transmitters to be put to sleep in-between samples and thereby save energy.

The research area of joint design of control and communication is currently
very active, especially in the context of event-triggered control. In
\cite{mol+hir09} a joint optimization of control and communication is solved
using dynamic programming, placing a communication scheduler in the sensor. In
\cite{ant+hee12:ACC}, \cite{ant+hee12:CDC} and \cite{hee+don:TAC13} the control
law and event-condition are co-designed to match performance of periodic
control using a lower communication rate. In \cite{hee+don:AUT13} this idea is
extended to decentralized systems. 
The use of predictive control is also
gaining popularity within the networked control community \cite{cas+06},
\cite{tan+sil06}, \cite{liu+06}, \cite{que+nes12}. In \cite{que+03} predictive
methods and vector quantization are used to reduce the controller to actuator
communication in multiple input systems. In \cite{lje+13} model predictive
control (MPC) is used to design multiple actuator link scheduling and control
signals. There have also been developments in using MPC under event-based
sampling. In \cite{ber+bem11} a method for trading control performance and
transmission rate in systems with multiple sensors is given. In \cite{leh+12}
an event-based MPC is proposed where the decision to re-calculate the control
law is based on the difference between predicted and measured states.

The problem addressed in the present paper, namely the joint design of a
self-triggering rule and the appropriate control signal using MPC has been less
studied than its event-triggered counterpart. In \cite{bar+12} an approach
relying on an exhaustive search which utilizes sub-optimal solutions giving the
control policy and a corresponding self-triggering policy is presented. In
\cite{eqt+13} it is suggested that a portion of the open loop trajectory
produced by the MPC should be applied to the process. The time between
re-optimizations is then decided via a self-triggering approach.

The approach taken in this paper differs from the above two in that the open
loop cost we propose the MPC to solve is designed to be used in an adaptive
sampling context. Further our extension to handle multiple loops using a
self-triggered MPC is new. So is the guarantee of conflict free transmissions.

The outline of the paper is as follows. In
Section~\ref{sec:problem_formulation} the self-triggered network scheduling and
control problem is defined and formulated as a receding horizon control
problem. Section~\ref{sec:open_loop_problem} presents the open-loop optimal
control problem for a single loop, to be solved by the receding horizon
controller. \edit{The optimal solution is presented in Section~\ref{sec:costFCNmin}}. Section~\ref{sec:rec_hor_imp}
presents the receding horizon control algorithm for a single loop in further
detail and gives conditions for when it is stabilizing. The results are then
extended to the multiple loop case in Section~\ref{sec:multiloop_extension}
where conditions for stability and conflict free transmissions are given. The
proposed method is explained and evaluated on simulated examples in
Section~\ref{sec:numerical_results}. Concluding discussions are made in
Section~\ref{sec:conclusions}.

\section{Self-Triggered Networked Control Architecture}\label{sec:problem_formulation}
We consider the problem of controlling $s\geq1$ processes $\mathcal{P}_1$
through $\mathcal{P}_s$ over a shared communication network as in
Fig.~\ref{fig:NCS}. The processes are controlled by the controller
$\mathcal{C}$ which computes the appropriate control action and schedule for
each process. Each process $\mathcal{P}_\ell$ is given by a linear
time-invariant (LTI) system
\begin{equation}\label{eq:multiprocess}
\begin{aligned}
    &x_\ell(k+1)=A_\ell x_\ell(k)+B_\ell u_\ell(k),\\
    &x_\ell(k)\in\mathbb{R}^{n_\ell},\,u_\ell(k)\in\mathbb{R}^{m_\ell}.
\end{aligned}
\end{equation}
The controller works in the following way: 
\edit{At time $k=k_\ell$, sensor $\mathcal{S}_\ell$ transmits a sample $x_\ell(k_\ell)$ to the controller, which then computes the control signal $u_\ell(k_\ell)$ and sends it to the actuator $\mathcal{A}_\ell$.}
\edit{Here $\ell \in \left\{1,2,\dots,s\right\}$ is the process index.}
The actuator in turn will apply this control
signal to the process until
\edit{ a new value is received from the controller.}
Jointly with deciding $u_\ell(k_\ell)$ the controller also decides how many
discrete time steps, say \edit{$I_\ell(k_\ell) \in \mathbb{N^+} \triangleq \{1,2,\dots \}$}, it will wait before it needs to
change the control signal the next time. This value $I_\ell(k_\ell)$ is sent to
the \emph{Network Manager} which will schedule the sensor $\mathcal{S}_\ell$ to
send a new sample at time $k=k_\ell+I_\ell(k_\ell)$.
\edit{To guarantee conflict free
transmissions on the network, only one sensor is allowed to transmit at every time instance.} Hence, when deciding the time to wait $I_\ell(k_\ell)$, the
controller must make sure that no other sensor $\mathcal{S}_q$, $q\neq\ell$,
already is scheduled for transmission at time $k=k_\ell+I_\ell(k_\ell)$. 

We propose that the controller $\mathcal{C}$ should be implemented as a
receding horizon controller which for an individual loop~$\ell$ at every
sampling instant $k=k_\ell$ solves an open-loop optimal control problem. It
does so by minimizing the infinite-horizon quadratic cost function
\begin{equation*}
\sum_{l=0}^{\infty}\bigg(\|x_\ell(k_\ell+l)\|^2_{Q_\ell} + \|u_\ell(k_\ell+l)\|^2_{R_\ell} \bigg)
\end{equation*}
\edit{subject to the user defined weights ${Q_\ell}$ and ${R_\ell}$, while taking
system dynamics into account. In Section~\ref{sec:open_loop_problem} we will embellish this cost function, such that control performance, inter sampling time and overall network
schedulability also are taken into consideration.}

\begin{remark}
\edit{
The focus is on a networked system where the network manager and controller are integrated in the same unit. This means, that the controller can send the schedules, that contain the transmission times of the sensors, directly to the network manager. This information needs to be transmitted to the sensor node such that it knows when to sample and transmit. For example, using the IEEE 802.15.4 superframe, described in \cite{ieee:802.15.4-2011:standard}, this can be done without creating additional overhead. Here the network manager broadcasts a beacon in the beginning of each superframe, which occurs at a fixed time interval. This beacon is received by all devices on the network and includes a schedule of the sensors that are allowed to transmit at given times.}
\end{remark}

\section{Adaptive Sampling Strategy}\label{sec:open_loop_problem}
For pedagogical ease we will in this section study the case when we control a
single process on the network allowing us to drop the loop index $\ell$. The
process we control has dynamics
\begin{equation}\label{eq:single_process}
\begin{aligned}
    x(k+1)=Ax(k)+Bu(k),\, x(k)\in\mathbb{R}^{n},\,u(k)\in\mathbb{R}^{m}
\end{aligned}
\end{equation}
and the open-loop cost function we propose the controller to minimize at every
sampling instant is
\begin{equation}\label{eq:J}
    J(x(k),i,\mathcal{U}) =  \frac{\alpha}{i} + \sum_{l=0}^{\infty}\bigg(\|x(k+l)\|^2_Q + \|u(k+l)\|^2_R \bigg)
\end{equation}
where $\alpha \in \mathbb{R^+}$ \editt{is a design variable that is used to trade off the cost of sampling against the cost of control and $i = I(k)\in \left\{1,2,\dots,p\right\}$ is the number of discrete time units to wait before taking the next sample, where $\rate \in \mathbb{N^+}$ is the maximum number of time units to wait until taking the next sample. Further, the design variables} $0<Q$ and $0<R$ are symmetric matrices of appropriate dimensions. We optimize this
cost over the constraint that the control sequence $\mathcal{U}\triangleq \{u(k),
u(k+1),\ldots\}$ should follow the specific shape illustrated in
Fig.~\ref{fig:predictionhorizon}, for some fixed period $\rate$. \edit{The period $\rate$ is the maximum amount of time steps the sensor is allowed to wait before taking the next sample. If $p = 1$, a sample is taken at every time instance and a regular receding horizon cost is obtained. If one on the other hand would select a large value of $p$, very long sampling intervals can be obtained. This can though affect the control performance significantly in case there are disturbances in the system that make the size of the state increase.}
\begin{figure}
    \centering
    \psfrag{0}[l][][\fragsize]{$x(k)$}
    \psfrag{1}[l][][\fragsize]{$x(k+i)$}
    \psfrag{2}[][][\fragsize]{$x(k+i+\rate)$}
    \psfrag{3}[][][\fragsize]{$x(k+i+2\cdot \rate)$}
    \psfrag{a}[][][\fragsize]{$u(k)$}
    \psfrag{b}[][][\fragsize]{$u(k+i)$}
    \psfrag{c}[][][\fragsize]{$u(k+i+\rate)$}
    \psfrag{d}[][][\fragsize]{$u(k+i+2\cdot \rate)$}
    \psfrag{l}[][][\fragsize]{$\ldots$}
    \includegraphics[width=0.9\columnwidth]{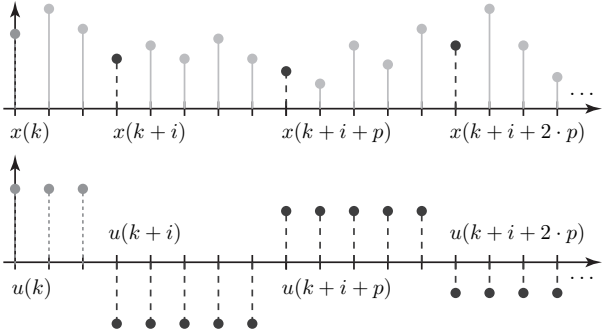}\\
    \caption{The prediction horizon. Here typical signal predictions are shown with $I(k)=3$ and $\rate=5$.}\label{fig:predictionhorizon}
\end{figure}
Thus, the number of discrete time units \editt{ $i=I(k) \in \left\{1,2,\dots,p\right\}$} to wait before taking the
next sample $x(k+i)$, as well as the levels in the control sequence
$\mathcal{U}$ are free variables over which we optimize. Note that neither the
state nor the control values have a constrained magnitude.

\edit{The constraint on the shape of the control trajectory $\mathcal{U}$ is motivated by the idea, that this sequence is applied to the actuator from time $k$ until time instant $k+i$. At time $x(k+i)$ we take a new sample and redo the optimization. By this method we get a joint optimization of the control signal to be applied as well as the number of time steps to the next sampling instant.} 
%
The reason for letting the system be controlled by a control
signal with period $\rate$ after this is that we hope for the receding horizon
algorithm to converge to this \edit{sampling-}rate. \edit{We will provide methods for choosing $\rate$ so that this happens in Section~\ref{sec:rec_hor_imp}.} The reason for
wanting convergence to a down sampled control is that we want the system to be
sampled at a slow rate when it has reached steady state, while we
want it to be sampled faster during the transients.


\edit{If one includes the above mentioned constraints, then (\ref{eq:J}) can be rewritten as}
\begin{equation}\label{eq:J_w_constraints}
\begin{aligned}
J(x(k),i,\mathcal{U}(i)) =&\frac{\alpha}{i}+ \sum_{l=0}^{i-1}\bigg(\|x(k+l)\|^2_Q + \|u(k)\|^2_R \bigg)\\
   +& \sum_{r=0}^\infty\Bigg(\sum_{l=0}^{\rate-1}\bigg(\|x(k+i+r\cdot \rate+l)\|^2_Q \\
   &\qquad+ \|u(k+i+r\cdot \rate)\|^2_R \bigg)\Bigg)\\
   \end{aligned}
\end{equation}
where \edit{$i$ and} $\mathcal{U}(i) \triangleq \{u(k), u(k+i), u(k+i+\eta \cdot \rate), \ldots\}$,
$\eta\in \mathbb{N}^+$, are the decision variables over which we optimize. The
term ${\alpha}/{i}$ reflects the cost of sampling. We use this cost to weight
the cost of sampling against the classical quadratic control performance cost.
For a given $x(k)$, choosing a large $\alpha$ will force $i$ to be larger and
hence give longer inter sampling times. 
\edit{By the construction of the cost, we may
tune $Q$, $R$ and $\alpha$ via simulations to get the desired sampling behaviour while maintaining the desired control performance}. One could imagine a more general
cost of sampling. \edit{Here however we found ${\alpha}/{i}$ sufficient to be able to trade control performance for communication cost, see simulations in Section~\ref{sec:compsingle}.} 

\section{Cost Function Minimization}\label{sec:costFCNmin} Having defined the open-loop cost
(\ref{eq:J_w_constraints}) we proceed by computing its optimal value. We start
by noticing that, even though we have a joint optimization problem, we may
state it as
\begin{equation}\label{eq:min_J_u}
\underset{i}{\text{minimize}} \Big( \underset{\mathcal{U}(i)}{\text{minimize}} \; J(x(k),i,\mathcal{U}(i)) \Big).
\end{equation}
We will use this separation and start by solving the inner problem, that of
minimizing $J(x(k),i,\mathcal{U}(i))$ for a given value of $i$. In order to
derive the solution, and for future reference, we need to define some
variables.
\begin{definition}\label{def:lifting}
We define notation for the lifted model as
\begin{equation*}
     A^\bi = A^{i}, \quad  B^\bi =\sum_{q=0}^{i-1}A^qB.
\end{equation*}
and notation for the generalized weighting matrices associated to
(\ref{eq:J_w_constraints}) as
\begin{equation*}
\begin{aligned}
     Q^\bi &=  Q^\bim +  {A^\bim}^T Q  A^\bim \\
     R^\bi &=  R^\bim +  {B^\bim}^T Q  B^\bim + R\\
     N^\bi &=  N^\bim +  {A^\bim}^T Q  B^\bim\\
\end{aligned}
\end{equation*}
where \edit{$i\in \left\{1,2,\dots,p\right\}$, }$Q^\bie = Q$, $R^\bie = R$ and $N^\bie = 0$.
\end{definition}

Using Definition~\ref{def:lifting} it is straightforward to show the following
lemma.
\begin{lemma}\label{lem:sum_collapse}
It holds that
\begin{equation*}
\begin{aligned}
&\sum_{l=0}^{i-1}\bigg(\|x(k+l)\|^2_Q + \|u(k)\|^2_R \bigg)\\
&= x(k)^TQ^\bi x(k)+ u(k)^TR^\bi u(k)+ 2x(k)^TN^\bi u(k)\\
    \end{aligned}
\end{equation*}
and $x(k+i) =  A^\bi  x(k)  +  B^\bi  u(k)$.
\end{lemma}

\begin{lemma}\label{lem:P_conversion}
Assume that $0 < Q$, $0<R$ and that the pair $(A^\bir,B^\bir)$ is controllable.
Then
\begin{multline*}
\min_{\mathcal{U}(i)}\sum_{r=0}^\infty\Bigg(\sum_{l=0}^{\rate-1}\bigg(\|x(k+i+r\cdot
\rate+l)\|^2_Q \\
+ \|u(k+i+r\cdot \rate)\|^2_R \bigg)\Bigg) = \|x(k+i)\|^2_{P^\bir}
\end{multline*}
and the minimizing control signal characterizing $\mathcal{U}(i))$ is given by
\begin{equation*}
u(k+i+r\cdot \rate) = -L^\bir x(k+i+r\cdot \rate).
\end{equation*}
In the above,
\begin{equation}\label{eq:ricc_periodic}
\begin{aligned}
P^\bir &=  Q^\bir +  {A^\bir}^TP^\bir A^\bir\\
&\qquad- \big( {A^\bir}^TP^\bir B^\bir+ N^\bir\big)L^\bir\\
L^\bir &=\big( R^\bir+{B^\bir}^TP^\bir B^\bir\big)^{-1}\\
&\qquad\times\big( {A^\bir}^TP^\bir B^\bir+ N^\bir\big)^T.\\
\end{aligned}
\end{equation}
\end{lemma}

\begin{proof}
\edit{The proof is given in Appendix~\ref{app:P_conversion}.}
\end{proof}

Using the above results we may formulate the main result of this section as
follows.
\begin{theorem}[Closed Form Solution]\label{thm:min_J}
Assume that $0 < Q$, $0<R$ and that the pair $(A^\bir,B^\bir)$ is controllable.
Then
\begin{equation}\label{eq:opt_open_loop_cost}
    \min_{\mathcal{U}(i)} J(x(k),i,\mathcal{U}(i)) =\frac{\alpha}{i}+ \|x(k)\|^2_{P^\bi}
\end{equation}
where
\begin{equation}\label{eq:ricc_i_step}
\begin{aligned}
P^\bi &=  Q^\bi +  {A^\bi}^TP^\bir A^\bi\\
&- \big( {A^\bi}^TP^\bir B^\bi+N^\bi\big)L^\bi\\
L^\bi &=\big(R^\bi+ {B^\bi}^TP^\bir B^\bi\big)^{-1}\\
&\times\big( {A^\bi}^TP^\bir B^\bi+ N^\bi\big)^T\\
\end{aligned}
\end{equation}
and $P^\bir$ is given by Lemma~\ref{lem:P_conversion}. Denoting the vector of
all ones in $\mathbb{R}^n$ as $\onevec_n$, the minimizing control signal
sequence is given by
\begin{equation*}
    \mathcal{U}^*=\{-L^\bi x(k)\onevec_i^T,-L^\bir x(k+i+r\cdot p)\onevec_p^T,\ldots\},\; r\in\mathbb{N}
\end{equation*}
where also $L^\bir$ is given by Lemma~\ref{lem:P_conversion}.
\end{theorem}

\begin{proof}
\edit{The proof is given in Appendix~\ref{app:min_J}.}
\end{proof}

Now, getting back to the original problem (\ref{eq:min_J_u}). Provided that the
assumptions of Theorem~\ref{thm:min_J} hold, we may apply it giving that
\begin{equation*}
\underset{i,\mathcal{U}(i)}{\text{min}} \; J(x(k),i,\mathcal{U}(i)) = \underset{i}{\text{min}} \; \left\{\frac{\alpha}{i}+ \|x(k)\|^2_{P^\bi}\right\}.
\end{equation*}
Unfortunately the authors are not aware of any method to solve this problem in
general. If however $i$ is restricted to a known \emph{finite} set
$\nomI\subset \mathbb{N}^+$ we may find the optimal value within this set for a
given value of $x(k)$ by simply evaluating $\frac{\alpha}{i}+
\|x(k)\|^2_{P^\bi}$ $\forall i \in \nomI$ and by this obtaining the $i$ which
gives the lowest value of the cost. This procedure gives the optimum of
(\ref{eq:min_J_u}). Note that the computational complexity of finding the
optimum is not necessarily high, as we may compute ${P^\bi}$ $\forall i \in
\nomI$ off-line prior to execution.

\section{Single Loop Self-Triggered MPC}\label{sec:rec_hor_imp} 
\edit{Having presented the receding horizon cost, we now continue with formulating its implementation in further detail.} We will assume that $0<Q$, $0<R$ and that the down-sampled pair
$(A^\bir,B^\bir)$ is controllable. Let us also assume that the \emph{finite}
and \emph{non-empty} set $\nomI\subset\mathbb{N}^+$ is given and let 
\begin{equation}\label{eq:maxI}
\maxI = \max\nomI.
\end{equation}
From this we may use Theorem~\ref{thm:min_J} to compute the pairs
$(P^\bi,L^\bi)$ $\forall i \in \nomI$ and formulate our proposed receding
horizon control algorithm for a single loop:
\begin{algorithm}\label{alg:singleloop}
Single Loop Self-Triggered MPC \hspace*{\fill}
\begin{enumerate}
  \item At time $k=k'$ the sample $x(k')$ is transmitted by the sensor to
      the controller.
  \item Using $x(k')$ the controller computes
\begin{equation*}
\begin{aligned}
I(k') &= \arg\min_{i\in\nomI}\frac{\alpha}{i}+ \|x(k')\|^2_{P^\bi}\\
u(k') &=-L^{(I(k'))}x(k')
\end{aligned}
\end{equation*}
\item {}
\begin{enumerate}
\item The controller sends $u(k')$ to the actuator which applies
    $u(k)=u(k')$ to (\ref{eq:single_process}) until $k=k'+I(k')$.
\item The network manager is requested to schedule the sensor to transmit a
    new sample at time $k=k'+I(k')$.
\end{enumerate}
\end{enumerate}
\end{algorithm}

\begin{remark}
Note that the control signal value is sent to the actuator at the same time as
the controller requests the scheduling of the next sample by the sensor.
\end{remark}

\begin{remark}
The proposed algorithm guarantees a minimum and a maximum inter-sampling time.
The minimum time is $1$ time step in the time scale of the underlying process
(\ref{eq:single_process}) and the maximum inter sampling time is $\maxI$ time
steps. This implies that there is some minimum attention to every loop
independent of the predicted evolution of the process.
\end{remark}

\begin{remark}
    Even though we are working in uniformly sampled discrete time the state is not
    sampled at every time instant $k$. Instead the set of samples
    of the state actually taken is given by the set
    $\mathcal{D}$, which assuming that the first sample is taken at $k=0$, is
    given by
    \begin{equation}\label{eq:set_of_samples}
        \mathcal{D} = \left\{x(0),x(I(0)), x(I(I(0))),\ldots \right\}.
    \end{equation}
\end{remark}

\section{Analysis} Having established and detailed our receding horizon
control law for a single loop we continue with giving conditions for when it is
stabilizing. Letting $\lambda(A)$ denote the set of eigenvalues of $A$ 
we first recall the following
controllability conditions.
\begin{lemma}\label{lem:downsample}\edit{\cite{colaneri1991regulation}}
The system $(A^\bi,B^\bi)$ is controllable if and only if the pair $(A,B)$ is
controllable and $A$ has no eigenvalue $\lambda\in\lambda(A)$ such that
$\lambda \neq 1$ and $\lambda^i=1$.
\end{lemma}

Using the above, and the results of Section \ref{sec:costFCNmin}, we may now give conditions for when
the proposed receding horizon control algorithm is stabilizing.
\begin{theorem}[Practical stability \cite{jiawan:2001:iss-dt}]\label{thm:stabilitybound}
Assume $0<Q$, $0<R$ and that $(A,B)$ is controllable. If we choose
$i\in\nomI\subset \mathbb{N}^+$ and $\rate=\optrate$ given by
\begin{equation}\label{eq:stab_rate_cond}
\optrate=\max \{i | i\in \nomI, \forall\lambda\in\lambda(A)\ \lambda^i \neq 1 \text{ if }\lambda \neq 1\}
\end{equation}
and apply Algorithm~\ref{alg:singleloop}, \edit{then the system state is ultimately bounded as per:}
\begin{equation*}
\frac{\alpha}{\maxI}\leq \lim_{k\rightarrow\infty}\min_{i\in\nomI}\bigg(\frac{\alpha}{i}+
\|x(k)\|^2_{P^\bi}\bigg) \leq
\frac{\alpha}{\epsilon}\bigg(\frac{1}{\optrate}-(1-\epsilon)\frac{1}{\maxI}\bigg).
\end{equation*}
\edit{In the above, $\maxI$ is as in (\ref{eq:maxI}) and $\epsilon$ is the largest value in the interval $(0,1]$ which ${\forall
i\in\nomI}$ fulfils}
\begin{equation*}
 \big(A^\bi - B^\bi L^\bi\big)^TP^\biro \big(A^\bi - B^\bi L^\bi\big) \leq (1-\epsilon)P^\bi
\end{equation*}
\edit{and is guaranteed to exist.}
\end{theorem}

\begin{proof}
\edit{The proof is given in Appendix~\ref{app:stabilitybound}.}
\end{proof}

\begin{remark}
The bound given in Theorem~\ref{thm:stabilitybound} scales linearly with the
choice of $\alpha$.
\end{remark}

\begin{assumption}\label{ass:pstarIsGamma}
Assume that $\nexists\lambda\in\lambda(A)$ except possibly $\lambda=1$ such
that $|\lambda|=1$ and the complex argument
$\angle\lambda=\frac{2\pi}{\maxI}\cdot n$ for some $n\in\mathbb{N}^+$.
\end{assumption}
\begin{lemma}\label{lem:pstarIsGamma}
Let Assumption~\ref{ass:pstarIsGamma} hold, then $\optrate=\maxI$.
\end{lemma}
\begin{proof}
\edit{The proof is given in Appendix~\ref{app:pstarIsGamma}.}
\end{proof}

\begin{corollary}[Asymptotic stability of the origin]\label{corr:stability}
Assume $0<Q$, $0<R$ and that $(A,B)$ is controllable. Further assume that
either Assumption~\ref{ass:pstarIsGamma} holds or $\alpha=0$. If we choose
$i\in\nomI\subset \mathbb{N}^+$ and $\rate=\optrate$ given by
(\ref{eq:stab_rate_cond}) and apply Algorithm~\ref{alg:singleloop}, then
\begin{equation*}
\lim_{k\rightarrow\infty}x(k) = 0
\end{equation*}
\end{corollary}

\begin{proof}
\edit{The proof is given in Appendix~\ref{app:stabilitycorr}.}
\end{proof}

From the above results we may note the following.
\begin{remark}
If the assumptions of Theorem~\ref{thm:stabilitybound} hold, \edit{then Corollary~\ref{corr:stability} can be used to ensure asymptotic stability of the origin,}
 except in the extremely rare case that
the underlying system $(A,B)$ becomes uncontrollable under down sampling by a
factor $\maxI$, see Lemma~\ref{lem:downsample}. Otherwise one may use
Lemma~\ref{lem:pstarIsGamma} to re-design $\nomI$ giving a new value of $\maxI$
which recovers the case $\optrate=\maxI$.
\end{remark}

\edit{
\begin{remark}
It is worth noticing that the developed framework is not limited to
just varying the time to the next sample $i$ as in
Fig.~\ref{fig:predictionhorizon}. One could expand this by using a ``multi-step'' approach wherein the inter sampling time is optimized over several frames before reverting to constant sampling with period $\rate$.
\end{remark}
}
\section{Extension to Multiple Loops}\label{sec:multiloop_extension}
Having detailed the controller for the case with a single loop on the network
and given conditions for when it is stabilizing we now continue with extending
to the case when we control multiple loops on the network, as described in
Fig.~\ref{fig:NCS}. The idea is that the controller $\mathcal{C}$ now will run
$s$ such single loop controllers described in Algorithm~\ref{alg:singleloop} in
parallel, one for each process $\mathcal{P}_\ell$,
$\ell\in\loopset=\{1,2,\ldots,s\}$, controlled over the network. To guarantee
conflict free communication on the network the controller $\mathcal{C}$ will,
at the same time, centrally coordinate the transmissions of the different
loops.

\subsection{Cost Function}
We start by extending the results in Section~\ref{sec:open_loop_problem} to the
case when we have multiple loops. The cost function we propose the controller
to minimize at every sampling instant for loop~$\ell$ is then
\begin{equation*}
\begin{aligned}
J_\ell(x_\ell(k),i,\mathcal{U}_\ell(i)) =&\frac{\alpha_\ell}{i}+ \sum_{l=0}^{i-1}\bigg(\|x_\ell(k+l)\|^2_{Q_\ell} + \|u_\ell(k)\|^2_{R_\ell} \bigg)\\
   +& \sum_{r=0}^\infty\Bigg(\sum_{l=0}^{\rate_\ell-1}\bigg(\|x_\ell(k+i+r\cdot \rate_\ell+l)\|^2_{Q_\ell} \\
   &\qquad+ \|u_\ell(k+i+r\cdot \rate_\ell)\|^2_{R_\ell} \bigg)\Bigg)\\
   \end{aligned}
\end{equation*}
derived in the same way as (\ref{eq:J_w_constraints}) now with $\alpha_\ell \in
\mathbb{R^+}$, $0<Q_\ell$ $0<R_\ell$ and period $\rate_\ell\in\mathbb{N}^+$
specific for the control of process $\mathcal{P}_\ell$ given by
(\ref{eq:multiprocess}). From this we can state the following.
\begin{definition}
We define the notation in the multiple loop case following
Definition~\ref{def:lifting}. For a matrix $E_\ell$, \eg, $A_\ell$ and
$Q_\ell$, we denote $\big(E_\ell\big)^\bi$ by $E^\bi_\ell$ .
\end{definition}
\begin{theorem}[Closed Form Solution]\label{thm:min_J_multi}
Assume that $0 < Q_\ell$, $0<R_\ell$ and that the pair
$(A^\bir_\ell,B^\bir_\ell)$ is controllable. Then
\begin{equation*}
    \min_{\mathcal{U}_\ell(i)} J_\ell(x_\ell(k),i,\mathcal{U}_\ell(i)) =\frac{\alpha_\ell}{i}+ \|x_\ell(k)\|^2_{P^\bi_\ell}
\end{equation*}
where
\begin{equation*}
\begin{aligned}
P^\bi_\ell &=  Q^\bi_\ell +  {A^\bi_\ell}^TP^\birl_\ell A^\bi_\ell- \big( {A^\bi_\ell}^TP^\birl_\ell B^\bi_\ell+
N^\bi_\ell\big)L^\bi_\ell\\
L^\bi_\ell &=\big(R^\bi_\ell+ {B^\bi_\ell}^TP^\birl_\ell B^\bi_\ell\big)^{-1}\big( {A^\bi_\ell}^TP^\birl_\ell B^\bi_\ell+ N^\bi_\ell\big)^T\\
\end{aligned}
\end{equation*}
and
\begin{equation*}
\begin{aligned}
P^\birl_\ell &=  Q^\birl_\ell +  {A^\birl_\ell}^TP^\birl_\ell A^\birl_\ell\\
&\qquad- \big( {A^\birl_\ell}^TP^\birl_\ell B^\birl_\ell+ N^\birl_\ell\big)L^\birl_\ell\\
L^\birl_\ell &=\big(R^\birl_\ell+ {B^\birl_\ell}^TP^\birl_\ell B^\birl_\ell\big)^{-1}\\
&\qquad\times\big( {A^\birl_\ell}^TP^\birl_\ell B^\birl_\ell+ N^\birl_\ell\big)^T.
\end{aligned}
\end{equation*}
Denoting the vector of all ones in $\mathbb{R}^{n}$ as $\onevec_{n}$, the
minimizing control signal sequence is given by
\begin{equation*}
    \mathcal{U}_\ell^*=\{-L^\bi_\ell x_\ell(k)\onevec_i^T,-L^\birl_\ell x_\ell(k+i+r\cdot p_\ell)\onevec_{p_\ell}^T,\ldots\},\; r\in\mathbb{N}.
\end{equation*}
\end{theorem}
\begin{proof}
Application of Theorem~\ref{thm:min_J} on the individual loops.
\end{proof}

\subsection{Multiple Loop Self-Triggered MPC}
To formulate our multiple loop receding horizon control law we will re-use the
results in Section~\ref{sec:rec_hor_imp} and apply them on a per loop basis.

For each process $\mathcal{P}_\ell$ with dynamics given in
(\ref{eq:multiprocess}) let the weights $\alpha_\ell \in \mathbb{R^+}$,
$0<Q_\ell$ $0<R_\ell$ and period $\rate_\ell\in\mathbb{N}^+$ specific to the
process be defined. If we further define the \emph{finite} set
$\nomI_\ell\subset\mathbb{N}^+$ for loop~$\ell$ we may apply
Theorem~\ref{thm:min_J_multi} to compute the pairs $(P^\bi_\ell,\,L^\bi_\ell)$
$\forall\,i\in\nomI_\ell$. Provided of course that the pair
$(A^\birl_\ell,A^\birl_\ell)$ is controllable. 

As discussed in Section~\ref{sec:problem_formulation} when
choosing $i$ the controller must take into consideration what transmission times other loops have
reserved as well as overall network schedulability. Hence at time $k=k_\ell$
loop~$\ell$ is restricted to choose
$i\in\mathcal{I}_\ell(k_\ell)\subseteq\nomI_\ell$ where
$\mathcal{I}_\ell(k_\ell)$ contains the feasible values of $i$ which gives
collisions free scheduling of the network. 
\edit{Note that at time $k=k_\ell$ the control input $u_\ell$ and time for the next sample $I_\ell$ only have to be computed for a \emph{single} loop~$\ell$.} How $\mathcal{I}_\ell(k_\ell)$
should be constructed when multiple loops are present on the network is
discussed further later in this section. 

We may now continue with formulating our proposed algorithm for controlling
multiple processes over the network. The following algorithm is executed
whenever a sample is received by the controller.
\begin{algorithm}\label{alg:multiloop}
Multiple Loop Self-Triggered MPC\hspace*{\fill}
\begin{enumerate}
  \item At time $k=k_\ell$ the sample $x_\ell(k_\ell)$ of process
      $\mathcal{P}_\ell$ is transmitted by the sensor $\mathcal{S_\ell}$ to
      the controller $\mathcal{C}$.
  \item The controller $\mathcal{C}$ constructs $\mathcal{I_\ell}(k_\ell)$.
  \item Using $x_\ell(k_\ell)$ the controller $\mathcal{C}$ computes
\begin{equation*}
\begin{aligned}
I_\ell(k_\ell) &= \arg\min_{i\in\mathcal{I_\ell}(k_\ell)}\frac{\alpha_\ell}{i}+ \|x_\ell(k_\ell)\|^2_{P^\bi_\ell}\\
u_\ell(k_\ell)&=-L_\ell^{(I(k_\ell))} x_\ell(k_\ell).
\end{aligned}
\end{equation*}
\item
\begin{enumerate}
\item The controller $\mathcal{C}$ sends $u_\ell(k_\ell)$ to the actuator
    $\mathcal{A_\ell}$ which applies $u_\ell(k)=u_\ell(k_\ell)$ to
    (\ref{eq:multiprocess}) until $k=k_\ell+I_\ell(k_\ell)$.
\item The \emph{Network Manager} is requested to schedule the sensor
    $\mathcal{S_\ell}$ to transmit a new sample at time
    $k=k_\ell+I_\ell(k_\ell)$.
 \end{enumerate}
 \end{enumerate}
\end{algorithm}
\edit{Note that at time $k$ step 3 only needs to be performed for loop $l$.} 

\edit{\begin{remark}\label{rem:multipleinit}
When the controller is initialized at time $k=0$ it is assumed that the
controller has knowledge of the state $x_\ell(0)$ for all processes
$\mathcal{P}_\ell$ controlled over the network. It will then execute
Algorithm~\ref{alg:multiloop} entering at step~$2$, in the order of increasing
loop index $\ell$.
\end{remark}}

\subsection{Schedulability}
What remains to be detailed in the multiple loop receding horizon control law
is a mechanism for loop~$\ell$ to choose $\mathcal{I}_\ell(k_\ell)$ to achieve
collision free scheduling. We now continue with giving conditions for when this
holds.

First we note that when using Theorem~\ref{thm:min_J_multi} we make the
implicit assumption that it is possible to apply the corresponding optimal
control signal sequence $\mathcal{U}^*_\ell$. For this to be possible we must
be able to measure the state $x_\ell(k)$ at the future time instances
\begin{multline}\label{eq:samplingpattern}
S_\ell(k_\ell)=\{k_\ell+I_\ell(k_\ell),\ k_\ell+I_\ell(k_\ell)+\rate_\ell,\\
k_\ell+I_\ell(k_\ell)+2\rate_\ell,\ k_\ell+I_\ell(k_\ell)+3\rate_\ell,\
\ldots\}.
\end{multline}
Hence this sampling pattern must be reserved for use by sensor
$\mathcal{S}_\ell$. We state the following to give conditions for when this is
possible.

\begin{lemma}\label{lem:persistent}
Let loop~$\ell$ choose its set $\mathcal{I}_\ell(k_\ell)$ of feasible times to
wait until the next sample to be
\begin{multline*}
\mathcal{I}_\ell(k_\ell)=\{i\in\nomI_\ell|i\neq k_q^\text{next}-k_\ell+n\cdot\rate_q-m\cdot\rate_\ell,\\
 m,n\in\mathbb{N},\, q\in\loopset\setminus\{\ell\}\}
\end{multline*}
where $k_q^\text{next}$ is the next transmission time of sensor
$\mathcal{S}_q$. Then it is possible to reserve the needed sampling pattern
$S_\ell(k_\ell)$ in (\ref{eq:samplingpattern}) at time $k=k_\ell$.
\end{lemma}

\begin{proof}
\edit{The proof is given in Appendix~\ref{app:persistent}.}
\end{proof}

Constructing $\mathcal{I}_\ell(k_\ell)$ as above we are not guaranteed that
$\mathcal{I}_\ell(k_\ell)\neq\emptyset$. To guarantee this we make the
following assumption.

\begin{assumption}\label{ass:alwayspersistent}
Assume that for every loop~$\ell$ on the network $\nomI_\ell=\nomI$ and
$\rate_\ell=\rate$. Further assume that $\loopset\subseteq\nomI$ and
$\max\loopset\leq\rate$.
\end{assumption}

\begin{theorem}\label{thm:alwayspersistent}
Let Assumption~\ref{ass:alwayspersistent} hold. If every loop~$\ell$ chooses
\begin{equation*}
\begin{aligned}
\mathcal{I}_\ell(k_\ell)&=\{i\in\nomI|i\neq k_q^\text{next}-k_\ell+r\cdot \rate,\,r\in\mathbb{Z},\, q\in\loopset\setminus\{\ell\} \}
\end{aligned}
\end{equation*}
all transmissions on the network will be conflict free and it will always be
possible to reserve the needed sampling pattern $S_\ell(k_\ell)$ in
(\ref{eq:samplingpattern}).
\end{theorem}

\begin{proof}
\edit{The proof is given in Appendix~\ref{app:alwayspersistent}.}
\end{proof}

\begin{remark}
The result in Lemma~\ref{lem:persistent} requires the reservation of an
infinite sequence. This is no longer required in
Theorem~\ref{thm:alwayspersistent} as all loops cooperate when choosing the set
of feasible times to wait. In fact loop~$\ell$ only needs to know the current
time $k_\ell$, the period $\rate$ and the times when the other loops will
transmit next $k_q^\text{next}$ $\forall\,q\in\loopset\setminus\{\ell\}$ in
order to find its own value $\mathcal{I}_\ell(k_\ell)$.
\end{remark}

\begin{remark}\label{rem:pstarInNext}
If Assumption~\ref{ass:alwayspersistent} holds and every loop on the network
chooses $\mathcal{I}_\ell(k_\ell)$ according to
Theorem~\ref{thm:alwayspersistent}, then it is guaranteed that at time $k_\ell$
we can reserve (\ref{eq:samplingpattern}) and that no other loop can make
conflicting reservations. Hence at time $k_\ell+I_\ell(k_\ell)$ the sequence
\begin{multline*}
S_\ell(k_\ell+I_\ell(k_\ell))=\{\ k_\ell+I_\ell(k_\ell)+\rate,\ \\
k_\ell+I_\ell(k_\ell)+2\rate,\ k_\ell+I_\ell(k_\ell)+3\rate,\ldots\}.
\end{multline*}
is guaranteed to be available. Thus
$\rate\in\mathcal{I}_\ell(k_\ell+I_\ell(k_\ell))$.
\end{remark}

\subsection{Stability}
We continue with giving conditions for when the multiple loop receding horizon
control law described in Algorithm~\ref{alg:multiloop} is stabilizing.
Extending the theory developed Section~\ref{sec:rec_hor_imp} to the multiple
loop case we may state the following.

\begin{theorem}\label{thm:multistabilitybound}
Assume $0<Q_\ell$, $0<R_\ell$ and that $(A_\ell,B_\ell)$ is controllable.
Further let Assumption~\ref{ass:alwayspersistent} hold. If we then choose
$i\in\mathcal{I}_\ell(k)\subseteq\nomI\subset \mathbb{N}^+$, with
$\mathcal{I}_\ell(k)$ chosen as in Theorem~\ref{thm:alwayspersistent}, and
$\rate=\optrate$ given by
\begin{equation}\label{eq:stab_rate_cond_multi}
\optrate=\max \{i | i\in \nomI,\forall\,\ell\ \forall\lambda\in\lambda(A_\ell)\ \lambda^i \neq 1 \text{ if }\lambda \neq 1\}
\end{equation}
and apply Algorithm~\ref{alg:multiloop}, then as $k\rightarrow\infty$
\begin{equation*}
\frac{\alpha_\ell}{\maxI}\leq \min_{i\in\mathcal{I}_\ell(k)}\bigg(\frac{\alpha_\ell}{i}+
\|x_\ell(k)\|^2_{P^\bi_\ell}\bigg) \leq
\frac{\alpha_\ell}{\epsilon_\ell}\bigg(\frac{1}{\optrate}-(1-\epsilon_\ell)\frac{1}{\maxI}\bigg)
\end{equation*}
where $\maxI=\max\nomI$ and $\epsilon_\ell$ is the largest value in the
interval $(0,1]$ which ${\forall i\in\nomI}$ fulfills
\begin{equation*}
 \big(A^\bi_\ell - B^\bi_\ell L^\bi_\ell\big)^TP^\biro_\ell\big(A^\bi_\ell - B^\bi_\ell L^\bi_\ell\big) \leq (1-\epsilon_\ell)P^\bi_\ell.
\end{equation*}
\end{theorem}
\begin{proof}
\edit{The proof is given in Appendix~\ref{app:multistabilitybound}.}
\end{proof}

\begin{corollary}\label{corr:multistability}
Assume $0<Q_\ell$, $0<R_\ell$ and that $(A_\ell,B_\ell)$ is controllable.
Further let Assumption~\ref{ass:alwayspersistent} hold. In addition let $\nomI$
be chosen so that the resulting $\maxI=\max\nomI$ guarantees that
Assumption~\ref{ass:pstarIsGamma} holds for every loop~$\ell$ or alternatively
let $\alpha_\ell=0$ for every loop~$\ell$. If we then choose
$i\in\mathcal{I}_\ell(k)\subseteq\nomI\subset \mathbb{N}^+$, with
$\mathcal{I}_\ell(k)$ chosen as in Theorem~\ref{thm:alwayspersistent}, and
$\rate=\optrate$ given by (\ref{eq:stab_rate_cond_multi}) and apply
Algorithm~\ref{alg:multiloop} it holds that
\begin{equation*}
 \lim_{k\rightarrow\infty} x_\ell(k)=0.
\end{equation*}
\end{corollary}
\begin{proof}
The proof follows from the results in Theorem~\ref{thm:multistabilitybound}
analogous to the proof of Corollary~\ref{corr:stability}.
\end{proof}

\section{Simulation Studies}\label{sec:numerical_results}
To illustrate the proposed theory we now continue with giving simulations.
First we show how the control law works when a single loop is controlled over
the network and focus on the loop specific mechanisms of the controller.
Secondly we illustrate how the controller works when multiple loops are present
on the network and focus on how the controller allocates network access to
different loops.

\subsection{Single Loop}\label{sec:examples}
Let us exemplify and discuss how the controller handles the control performance
versus communication rate trade-off in an individual loop. We do this by
studying the case with a single system on the network. The system we study is
the single integrator system which we discretize using sample and hold with
sampling time $T_s=1$\,s giving us $x(k+1)=Ax(k)+Bu(k)$ with $(A,B)=(1,1)$.
Since we want the resulting self-triggered MPC described in
Algorithm~\ref{alg:singleloop} to be stabilizing we need to make sure that our
design fulfils the conditions of Theorem~\ref{thm:stabilitybound}. If we
further want it to be asymptotically stabilizing we in addition need it to
fulfil the conditions of Corollary~\ref{corr:stability}.

The design procedure is then as follows: First we note that the system $(A,B)$
is controllable. The next step is to decide the weights $0<Q$ and $0<R$ in the
quadratic cost function (\ref{eq:J}). This is done in the same way as in
classical linear quadratic control, see e.g. \cite{mac02}. Here we for
simplicity choose $Q=1$ and $R=1$. We note that the system only has the
eigenvalue $\lambda=1$, fulfilling Assumption~\ref{ass:pstarIsGamma}, so that
(\ref{eq:stab_rate_cond}) in Theorem~\ref{thm:stabilitybound} gives
$\optrate=\max \nomI$. Hence Corollary~\ref{corr:stability}, and thus
Theorem~\ref{thm:stabilitybound}, will hold for every choice of $\nomI$. This
means that we may choose the elements in $\nomI$, \ie, the possible down
sampling rates, freely. A natural way to choose them is to decide on a maximum
allowed down sampling rate and then choose $\nomI$ to contain all rates from
$1$ up to this number. Let's say that we here want the system to be sampled at
least every $5\cdot T_s$\,s, then a good choice is $\nomI=\{1,2,3,4,5\}$,
giving $\optrate=\max \nomI=5$.

Now having guaranteed that the conditions of Theorem~\ref{thm:stabilitybound}
and Corollary~\ref{corr:stability} hold we have also guaranteed that the
conditions of Theorem~\ref{thm:min_J} are fulfilled. Hence we may use it to
compute the state feedback gains and cost function matrices that are used in
Algorithm~\ref{alg:singleloop}. The results from these computations are shown
in Table~\ref{tab:num_int_ex}, together with some of the intermediate variables
from Definition~\ref{def:lifting}.
\begin{table}
\setlength{\tabcolsep}{4pt} \setlength{\extrarowheight}{2pt} \centering
\caption{The pre-computed control laws with related cost functions and
intermediate variables.}\label{tab:num_int_ex}
\begin{tabular}{c|cc|ccc|ccc}
  $i$ & $A^\bi$ & $B^\bi$ & $Q^\bi$ & $R^\bi$ & $N^\bi$ & $L^\bi$ & $P^\bi$ & $V^\bi(x)$ \\
  \hline
  1 & 1 & 1 & 1 & 1  & 0  & 0.70 & 1.70 & $\alpha/1 + P^{(1)}\cdot x^2$ \\
  2 & 1 & 2 & 2 & 3  & 1  & 0.46 & 1.73 & $\alpha/2 + P^{(2)}\cdot x^2$ \\
  3 & 1 & 3 & 3 & 8  & 3  & 0.35 & 1.89 & $\alpha/3 + P^{(3)}\cdot x^2$ \\
  4 & 1 & 4 & 4 & 18 & 6  & 0.28 & 2.08 & $\alpha/4 + P^{(4)}\cdot x^2$ \\
  5=$\optrate$ & 1 & 5 & 5 & 35 & 10 & 0.23 & 2.30 & $\alpha/5 + P^{(5)}\cdot x^2$
\end{tabular}
\end{table}
Here we see that the cost functions are quadratic functions in the state $x$
where the coefficients $P^\bi$ are functions of $Q$ and $R$. We also see that
the cost to sample $\alpha/i$ enters linearly and as we change it we will
change the offset level of the curves and thereby their values related to each
other. However it will not affect the state feedback gains.

A graphical illustration of the cost functions in Table~\ref{tab:num_int_ex},
for the choice $\alpha=0.2$, is shown in Fig.~\ref{fig:statespace_partition}
together with the curve $I(k)=\arg\min_{\nomI}V^\bi(x(k))$, \ie, the index of
the cost function which has the lowest value for a given state $x(k)$.
\begin{figure}
    \centering
    \psfrag{v}[][][1]{}
    \psfrag{x}[][][1]{$x(k)$}
    \psfrag{a}[][r][0.5]{$V^{(1)}$}
    \psfrag{b}[][r][0.5]{$V^{(2)}$}
    \psfrag{c}[][r][0.5]{$V^{(3)}$}
    \psfrag{d}[][r][0.5]{$V^{(4)}$}
    \psfrag{e}[][r][0.5]{$V^{(5)}$}
    \psfrag{k}[][][0.5]{$\arg\min_{\mathcal{I}(k)}V^{(i)}(x(k))$}
    \includegraphics[width=0.39\textwidth]{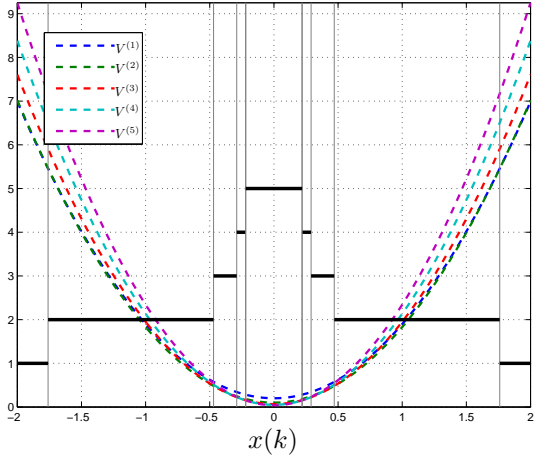}
    \caption{The cost functions $V^\bi(x(k))$ (dashed) together with the partitioning of the state space and the time to wait $I(k)=\arg\min_{\nomI}V^\bi(x(k))$ (solid).}
    \label{fig:statespace_partition}
\end{figure}
This is the partitioning of the state space that the self-triggered MPC
controller will use to choose which of the state feedback gains to apply and
how long to wait before sampling again.

Applying our self-triggered MPC described in Algorithm~\ref{alg:singleloop}
using the results in Table~\ref{tab:num_int_ex} to our integrator system when
initialized in $x(0)=2$ we get the response shown in
Fig.~\ref{fig:scalar_adap}.
Note here that the system will converge to the fixed
sampling rate $\optrate$ as the state converges.

It may now appear as it is sufficient to use periodic control and sample the
system every $\optrate\cdot T_s$\,s to get good control performance. To compare
the performance of this periodic sampling strategy with the self-triggered
strategy above we apply the control which minimizes the same cost function
(\ref{eq:J}) as above with the exception that the system now may only be
sampled every $\optrate\cdot T_s$\,s. This is in fact the same as using the
receding horizon control above while restricting the controller to choose
$i=\optrate$ every time. The resulting simulations are shown in
Fig.~\ref{fig:scalar_slow}. As seen, there is a large degradation of the
performance in the transient while the stationary behaviour is almost the same.
By this we can conclude that it is not sufficient to sample the system every
$\optrate\cdot T_s$\,s if we want to achieve the same transient performance as
with the self-triggered sampling.

In the initial transient response the self-triggered MPC controller sampled
after one time instant. This indicates that there is performance to gain by
sampling every time instant. To investigate this we apply the control which
minimizes the same cost function (\ref{eq:J}), now with the exception that the
system may be sampled every $T_s$\,s, \ie, classical unconstrained linear
quadratic control. Now simulating the system we get the response shown in
Fig.~\ref{fig:scalar_fast}. As expected we get slightly better transient
performance in this case compared to our self-triggered sampling scheme, it is
however comparable. Note however that this improvement comes at the cost of a
drastically increased communication need, which may not be suitable for systems
where multiple loops share the same wireless medium.
\begin{figure}
\centering
    \psfrag{k}[][][.9]{$k$}
    \psfrag{y}[][][.9]{$x(k)$}
    \psfrag{u}[][][.9]{$u(k)$}
    \subfigure[System response of the integrator system when minimizing the cost by using our single loop self-triggered MPC.]{
        \includegraphics[width=0.39\textwidth]{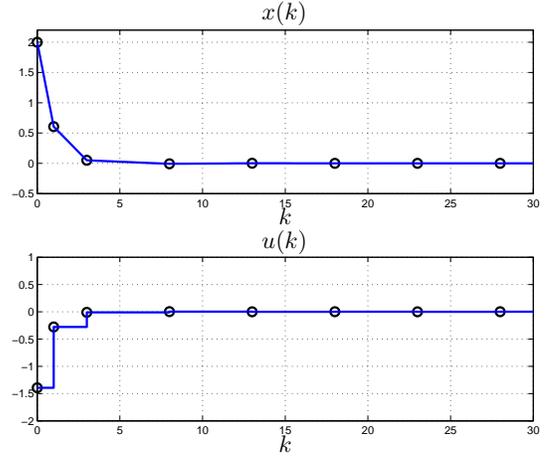}
        \label{fig:scalar_adap}
    }
    \subfigure[System response of the integrator system when minimizing the cost by sampling every $5^{th}$ second.]{
        \includegraphics[width=0.39\textwidth]{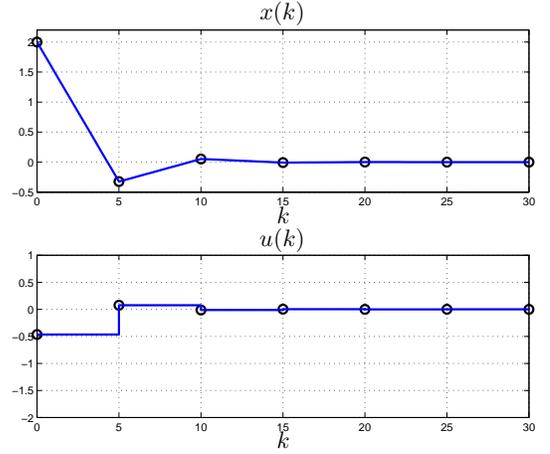}
        \label{fig:scalar_slow}
    }
        \subfigure[System response of the integrator system when minimizing the cost by sampling every second.]{
        \includegraphics[width=0.39\textwidth]{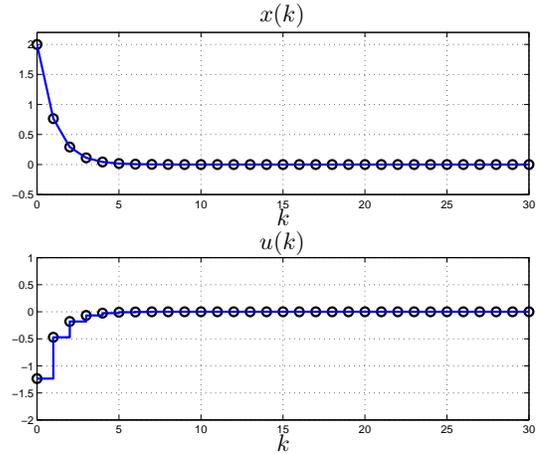}
        \label{fig:scalar_fast}
    }
    \caption{Comparing control performance for different sampling policies.}
\end{figure}

From the above we may conclude that our self-triggered MPC combines the low
communication rate in stationarity of the slow periodic controller with the
quick transient response of the fast periodic sampling. In fact we may, using
our method, recover the transient behaviour of fast periodic sampling at the
communication cost of one extra sample compared to slow periodic sampling. The
reason for this is that the fast sampling rate only is needed in the transient
while we in stationarity can obtain sufficient performance with a lower rate.

\subsection{Comparison to periodic control}\label{sec:compsingle}
To make a more extensive comparison between the self-triggered algorithm and periodic sampling, we sweep the value of the design parameter $\alpha$ in the cost function, and run 100 simulations of 10\,000 samples for each value. To add randomness to the process, the model (\ref{eq:multiprocess}) is extended to include disturbances such that for process $\ell$
\begin{equation}\label{eq:processnoise}
\begin{aligned}
    &x_\ell(k+1)=A_\ell x_\ell(k)+B_\ell u_\ell(k) + E_\ell \omega_\ell(k),
\end{aligned}
\end{equation}
where $\omega_\ell(k) \sim \mathcal{N}(0,\sigma_\ell^2)$ is zero-mean normal distributed with variance $\sigma_\ell^2 = 0.1$ and $E_\ell = 1$. We increase the maximum allowed sampling interval to $p^* = \max \mathcal{I}^0 = 15$. This though has the effect, that $P^{(1)} = 1.83$ whereas $P^{(2)} = 1.74$, such that it is always more favourable to have a sample interval of 2 samples instead of 1, even when $\alpha = 0$. We therefore changed the cost of control to $R=0.1$ such that we allow for larger control signals and force the process to sample at every time instance if $\alpha = 0$ is chosen. 
The initial state is randomly generated as $x_\ell(0) \sim \mathcal{N}(0,\sigma_{x_0,\ell}^2)$ , where $\sigma_{x_0,\ell}^2 = 2\,500$. 


%
\begin{figure}[tp]
    \centering
    \psfrag{SI}[c][][0.8]{Average sampling interval}
    \psfrag{a}[c][][0.8]{$\alpha$}
    \includegraphics[width=0.39\textwidth]{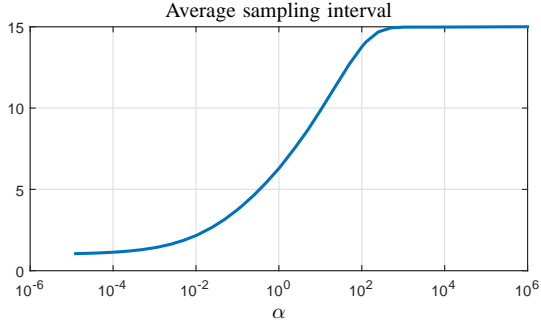}
    \caption{The average sampling interval for different values of $\alpha$ for the integrator system.}
    \label{fig:alpha-sample-interval}
\end{figure}

Fig.~\ref{fig:alpha-sample-interval} shows the average sampling interval during the simulations for different choices of $\alpha$ for the integrator system. The relation between the value of $\alpha$ and the average sampling interval is in general monotonic but highly affected by the system to be controlled and the statistics of the disturbances in the system.

Since we now have the range of values in which we want to sweep $\alpha$, we present a simple periodic MPC cost to which we compare our algorithm. This cost is given by
\begin{multline*}
	J_\ell(x_\ell(k),\mathcal{U}_\ell) = \sum_{r=0}^\infty \sum_{l=0}^{T_s} \left( \|x_\ell(k+l+r T_s) \|_{Q_\ell}^2 \right.\\
	\left.+ \| u_\ell(k+l+r T_s) \|_{R_\ell}^2 \right),
\end{multline*}
which using Definition~\ref{def:lifting} and Lemmas~\ref{lem:sum_collapse} and \ref{lem:P_conversion} can be rewritten as
%
\begin{multline}\label{eq:costperiodic}
	J_\ell(x_\ell(k),\mathcal{U}_\ell) = \sum_{r=0}^\infty \|x_\ell(k+r T_s)\|_{Q_\ell^{(T_s)}}^2 + \|u_\ell(k+r T_s)\|_{R_\ell^{(T_s)}}^2\\
	 + 2x_\ell(k+r T_s)^T N_\ell^{(i)} u_\ell(k+r T_s)
\end{multline}
for sampling period $T_s$ where
\begin{equation*}
	x_\ell(k+T_s) = A_\ell^{(T_s)}x_\ell(k) + B_\ell^{(T_s)}u_\ell(k).
\end{equation*}
When $T_s = 1$, (\ref{eq:costperiodic}) reduces to classical unconstrained linear quadratic control.

The empiric cost for each simulation is calculated by 
\begin{equation}\label{eq:empcost}
 \frac{1}{T} \sum_{k=0}^{T-1} x_\ell(k)^T Q_\ell x_\ell(k) + u_\ell(k)^T R_\ell u_\ell(k),
\end{equation}
where $T = 10\,000$ is the length of the simulation. 

\begin{figure}[tp]
    \centering
    \psfrag{S}[l][][0.8]{Self-triggered}
    \psfrag{P}[l][][0.8]{Periodic}
    \psfrag{AS}[c][][0.8]{Average sampling interval}
    \psfrag{IC}[c][][0.8]{Average empiric cost}
\psfrag{a1}[l][][0.6]{ $\alpha = 0$}    
\psfrag{a2}[l][][0.6]{ $\alpha = 0.05$} 
\psfrag{a3}[l][][0.6]{ $\alpha = 0.25$} 
\psfrag{a4}[l][][0.6]{ $\alpha = 1.3$}  
\psfrag{a5}[l][][0.6]{ $\alpha = 5$}    
\psfrag{a6}[l][][0.6]{ $\alpha = 25$}   
\psfrag{a7}[l][][0.6]{ $\alpha = 50$}   
\psfrag{a8}[l][][0.6]{ $\alpha = 5\cdot 10^2$}
\psfrag{a9}[l][][0.6]{ $\alpha = 1\cdot 10^6$}
    \includegraphics[width=0.39\textwidth]{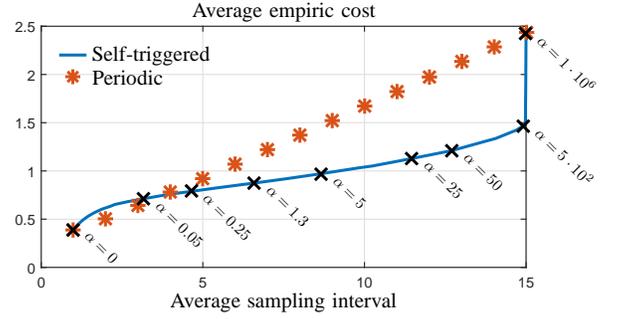}
    \caption{Performance of the integrator using the self-triggered algorithm compared to the simple periodic algorithm for different sampling intervals. 
    Some values of $\alpha$ are marked with the black crosses.}
    \label{fig:self-vs-periodic-single}
\end{figure}

Fig.~\ref{fig:self-vs-periodic-single} shows the average performance that is obtained for different averaged sampling intervals obtained by sweeping $\alpha$ from 0 to $10^6$.
Fig.~\ref{fig:self-vs-periodic-single} illustrates that the self-triggered algorithm performs significantly better than periodic sampling especially when communication costs are medium to high. However, at sampling intervals of 2 and 3 time-steps the self-triggered algorithm performs slightly worse than the periodic algorithm. This is caused by the fact, that if the state of the system is close to zero when sampled, the cost of sampling is much higher than the cost of the state error and control, hence the time until the next sample is taken is large. To avoid this phenomenon, in future work, one could take the statistics of the process noise into account in the cost function. It can further be noticed, that when $\alpha \rightarrow \infty$, the performance of the self-triggered algorithm is very close to the cost of the periodic algorithm. This is as expected, since the cost of $\alpha$ will be greater than the cost of the state and control, which will result in sample intervals of $i=p^* = 15$. This reduces the self-triggered algorithm to periodic control.

\subsection{Multiple Loops} We now continue with
performing a simulation study where we control two systems over the same
network. \edit{We will start by showing a simple example followed by a more extensive performance comparison of the Algorithm~\ref{alg:multiloop} to periodic control}. We will keep the integrator system from Section~\ref{sec:examples} now
denoting it process $\mathcal{P}_1$ with dynamics
$x_1(k+1)=A_1x_1(k)+B_1u_1(k)$ with $(A_1,B_1)=(1,1)$ as before. In addition we
will the control process $\mathcal{P}_2$ which is a double integrator system
which we discretize using sample and hold with sampling time $T_s=1$\,s giving
\begin{equation*}
\begin{aligned}
\underbrace{\left(\begin{array}{c}
           x^1_2(k+1) \\
           x^2_2(k+1)
         \end{array}\right)}_{x_2(k+1)}
         &=
          \underbrace{\left(\begin{array}{cc}
           1 & 0 \\
           1 & 1
         \end{array}\right)}_{A_2}
\underbrace{\left(\begin{array}{c}
           x^1_2(k) \\
           x^2_2(k)
         \end{array}\right)}_{x_2(k)}
+
    \underbrace{\left(\begin{array}{c}
           1 \\
           0.5
         \end{array}\right)}_{B_2}
        u_2(k).
\end{aligned}
\end{equation*}

We wish to control these processes using our proposed multiple loop
self-triggered MPC described in Algorithm~\ref{alg:multiloop}. As we wish to
stabilize these systems we start by checking the conditions of
Theorem~\ref{thm:multistabilitybound} and Corollary~\ref{corr:multistability}.
First we may easily verify that both the pairs $(A_1,B_1)$ and $(A_2,B_2)$ are
controllable. To use the stability results we need
Assumption~\ref{ass:alwayspersistent} to hold, implying that we must choose
$\rate_1=\rate_2=\rate$, $\nomI_1=\nomI_2=\nomI$ and choose $\nomI$ such that
$\{1,2\}\in\nomI$ and $2\leq \rate$. For reasons of performance we wish to
guarantee that the systems are sampled at least every $5\cdot T_s$\,s and
therefore choose $\nomI_1=\nomI_2=\nomI=\{1,2,3,4,5\}$ fulfilling the
requirement above. We also note that $\lambda(A_1)=\{1\}$ and
$\lambda(A_2)=\{1,1\}$ and that hence both system fulfill
Assumption~\ref{ass:pstarIsGamma} for this choice of $\nomI$, implying that
(\ref{eq:stab_rate_cond_multi}) in Theorem~\ref{thm:multistabilitybound} gives
$\optrate=\max\nomI=5$. Thus choosing $\rate=\optrate$ as stated in
Theorem~\ref{thm:multistabilitybound} results in that
Assumption~\ref{ass:alwayspersistent} holds. What now remains to be decided are
the weights $\alpha_\ell$, $Q_\ell$ and $R_\ell$.

For the integrator process $\mathcal{P}_1$ we keep the same tuning as in
Section~\ref{sec:examples} with $Q_1=R_1=1$. Having decided $Q_1$, $R_1$,
$\nomI$ and $\optrate$ we use Theorem~\ref{thm:min_J_multi} to compute the
needed state feedback gains and cost function matrices $(P^\bi_1,L^\bi_1)$
$\forall\,i\in\nomI$ needed by Algorithm~\ref{alg:multiloop}. We also keep
$\alpha_1=0.2$ as it gave a good communication versus performance trade-off.

For the double integrator process $\mathcal{P}_2$ the weights are chosen to be
$Q_2=I$ as we consider both states equally important and $R_2=\frac{1}{10}$ to
favor control performance and allow for larger control signals. Having decided
$Q_2$, $R_2$, $\nomI$ and $\optrate$ we may use Theorem~\ref{thm:min_J_multi}
to compute the needed state feedback gains and cost function matrices
$(P^\bi_2,L^\bi_2)$ $\forall\,i\in\nomI$ needed by
Algorithm~\ref{alg:multiloop}. The sampling cost is chosen to be $\alpha_2 =
1$, as this gives a good tradeoff between control performance and the number of
samples.

We have now fulfilled all the assumptions of both
Theorem~\ref{thm:multistabilitybound} and Corollary~\ref{corr:multistability}.
Hence applying Algorithm~\ref{alg:multiloop} choosing
$\mathcal{I}_\ell(k_\ell)$ according to Theorem~\ref{thm:alwayspersistent} will
asymptotically stabilize both process $\mathcal{P}_1$ and $\mathcal{P}_2$.

Controlling $\mathcal{P}_1$ and $\mathcal{P}_2$ using our multiple loop
self-triggered MPC described in Algorithm~\ref{alg:multiloop} with the above
designed tuning we get the result shown in Fig.~\ref{fig:dual_sys}.
\begin{figure}
\centering
    \psfrag{k}[][][.9]{$k$}
    \psfrag{y}[][][.9]{$x_1(k)$}
    \psfrag{u}[][][.9]{$u_1(k)$}
    \subfigure[System response for process $\mathcal{P}_1$]{
        \includegraphics[width=0.39\textwidth]{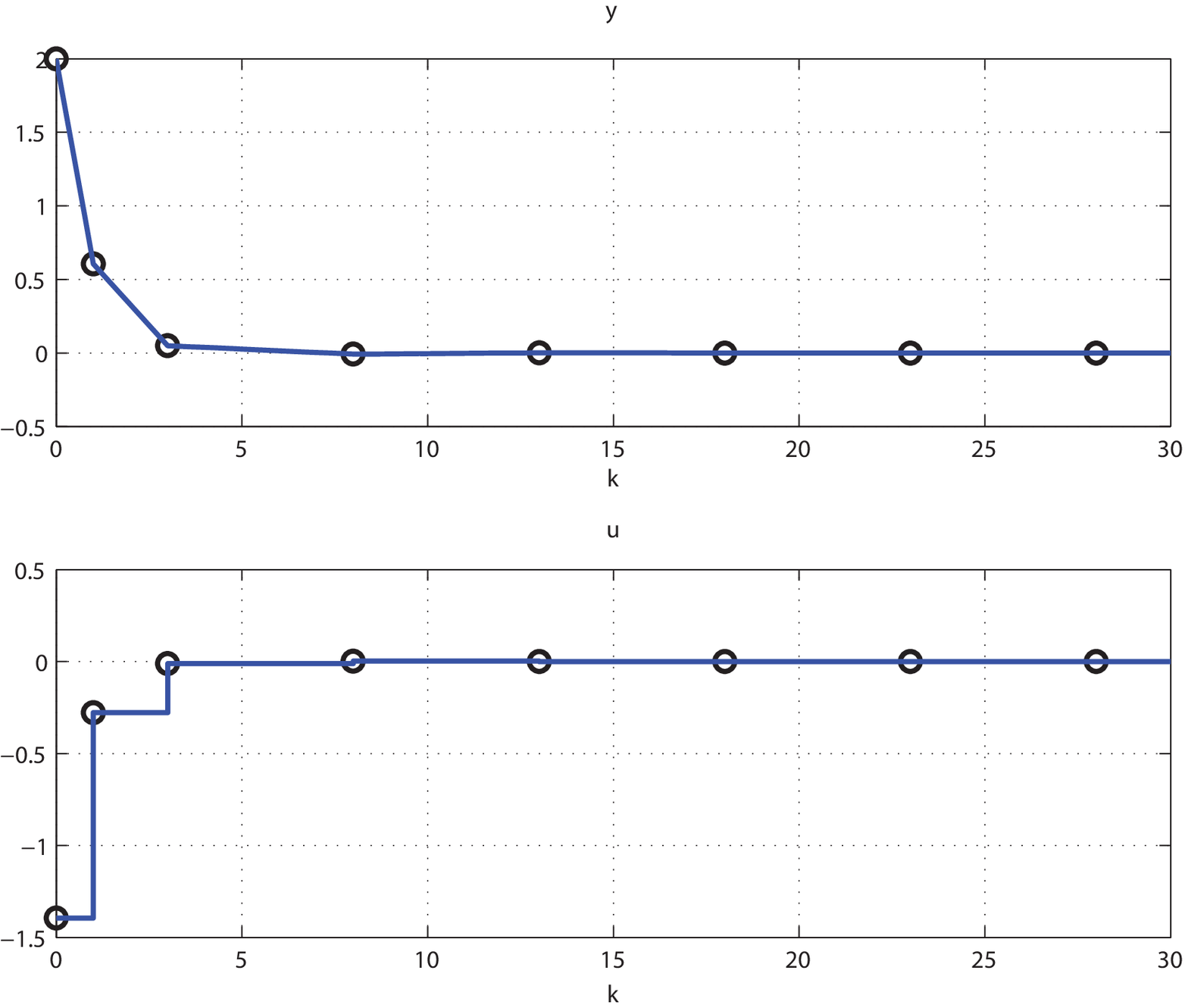}
        \label{fig:dual_sys1}
    }
    \psfrag{y}[][][.9]{$x_2(k)$}
    \psfrag{u}[][][.9]{$u_2(k)$}
    \subfigure[System response for process $\mathcal{P}_2$]{
        \includegraphics[width=0.39\textwidth]{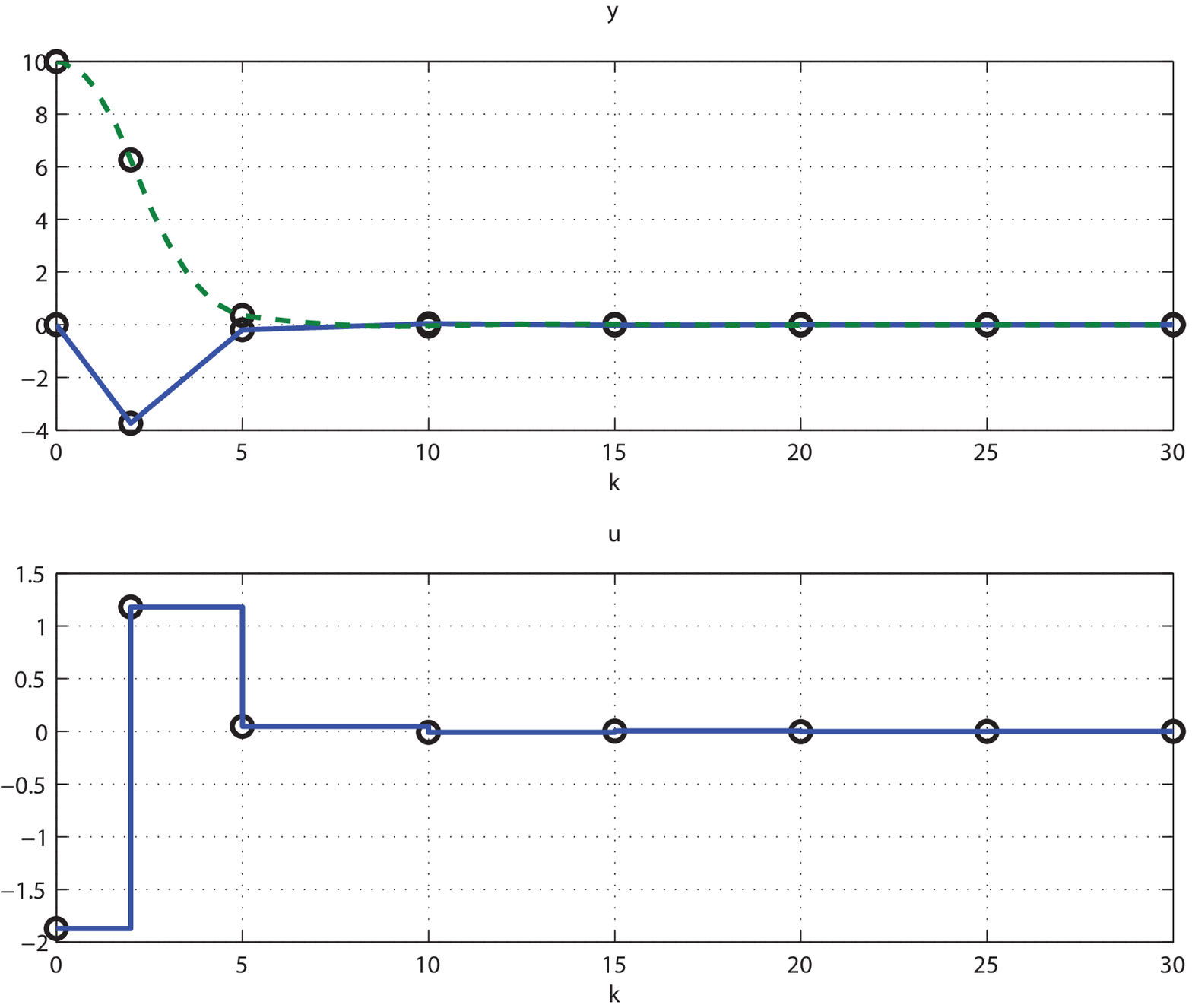}
    }
    \caption{The processes $\mathcal{P}_1$ and $\mathcal{P}_2$ controlled and scheduled on the same network using our multiple loop
self-triggered MPC.}
    \label{fig:dual_sys}
\end{figure}
As expected, the behavior of the controller illustrated in
Section~\ref{sec:examples} carries through also to the case when we have
multiple loops on the network. In fact comparing Fig.~\ref{fig:scalar_adap}
showing how the controller handles process $\mathcal{P}_1$ when controlling it
by itself on the network and Fig.~\ref{fig:dual_sys1} which shows how
$\mathcal{P}_1$ is handled in the multiple loop case we see that they are the
same. Further we see that, as expected, in stationarity the two loops
controlling process $\mathcal{P}_1$ and $\mathcal{P}_2$ both converge to the
sampling rate $\optrate$.

As mentioned previously the controller uses the mechanism in
Theorem~\ref{thm:alwayspersistent} to choose the set of feasible times to wait
until the next sample. In Fig.~\ref{fig:backofftimes} we can see how the
resulting sets $\mathcal{I}_\ell(k_\ell)$ look in detail. At time $k=0$
loop~$1$ gets to run Algorithm~\ref{alg:multiloop} first. As sensor
$\mathcal{S}_2$ is not scheduled for any transmissions yet
$\mathcal{I}_1(0)=\nomI$ from which the controller chooses $I_1(0)=1$. Then
loop~$2$ gets to run Algorithm~\ref{alg:multiloop} at time $k=0$. As sensor
$\mathcal{S}_1$ now is scheduled for transmission at time $k=0+I_1(0)=1$,
Theorem~\ref{thm:alwayspersistent} gives
$\mathcal{I}_2(0)=\nomI\setminus\{1\}$. From which the controller chooses
$I_2(0)=2$. The process is then repeated every time a sample is transmitted to
the controller, giving the result in Fig.~\ref{fig:backofftimes}.
\begin{figure}
\centering
    \psfrag{k}[][][.9]{$k$}
    \psfrag{y}[][][.9]{$\mathcal{I}_1(\cdot)$}
    \psfrag{z}[][][.9]{$\mathcal{I}_2(\cdot)$}
    \psfrag{i}[][][.9]{$\nomI$}
    \includegraphics[width=0.39\textwidth]{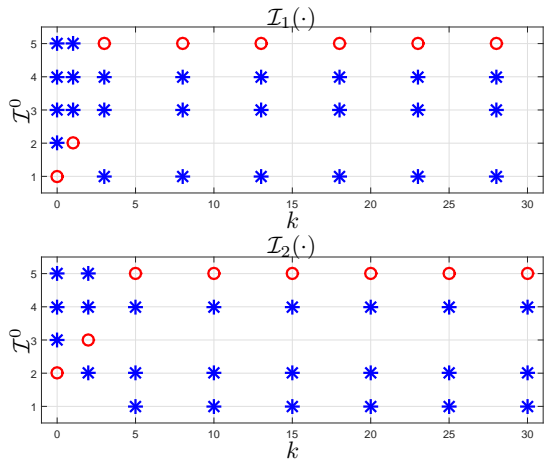}
    \caption{\edit{The sets $\mathcal{I}_\ell(\cdot)$ of feasible times to wait until the next sample for loop~$1$ and loop~$2$. The optimal time to wait $I_\ell(\cdot)$ is marked by red circles whereas the other feasible times are starred.}}
    \label{fig:backofftimes}
\end{figure}
As seen both the set $\mathcal{I}_\ell(\cdot)$ and the optimal time to wait
$I_\ell(\cdot)$ converges to some fixed value as the state of the corresponding
process $\mathcal{P}_\ell$ converges to zero.

\subsection{Comparison to periodic control}\label{sec:compmultiple}

For a more thorough performance comparison, we simulate the systems using Algorithm~\ref{alg:multiloop} and compare them to the periodic algorithm that uses the cost function (\ref{eq:costperiodic}). The single and double integrator processes ($\mathcal{P}_1$ and $\mathcal{P}_2$) are simulated using the parameters mentioned in Section~\ref{sec:numerical_results}. The variance of the disturbances for both processes are set to $\sigma_\ell^2 = 0.1,\,\forall \ell$, the variance of the initial state to $\sigma_{x_0,\ell}^2 = 25,\, \forall \ell$ and  $E_2 = [1,\ 1]^T$ in (\ref{eq:processnoise}). The value of $\alpha$ is identical for both processes in each simulation, such that $\alpha_1 = \alpha_2$. Further $R_1 = 0.1$. The simulation for both algorithms is initialized as described in Remark~\ref{rem:multipleinit}.

Fig~\ref{fig:self-vs-periodic-2N5} shows the performance for the processes, $\mathcal{P}_1$ and $\mathcal{P}_2$, calculated by averaging (\ref{eq:empcost}) over 100 simulations each of length 10\,000 for different values of $\alpha$ when $p_1^* = p_2^* = 5$. Fig.~\ref{fig:self-vs-periodic-2N15} shows the performance when $p_1^* = p_2^* = 15$. The cost of each process is normalized with respect to its largest cost for easier viewing. 

\begin{figure}[tp]
    \centering
    \psfrag{S1}[l][][0.8]{Self-triggered $\mathcal{P}_1$}
    \psfrag{S2}[l][][0.8]{Self-triggered $\mathcal{P}_2$}
    \psfrag{P1}[l][][0.8]{Periodic $\mathcal{P}_1$}
    \psfrag{P2}[l][][0.8]{Periodic $\mathcal{P}_2$}    
    \psfrag{AS}[c][][0.8]{Average sampling interval}
    \psfrag{IC}[c][][0.8]{Normalized average empiric cost}
\psfrag{a1}[l][][0.6]{ $\alpha = 0$}   
\psfrag{a2}[l][][0.6]{ $\alpha = 0.13$}
\psfrag{a3}[l][][0.6]{ $\alpha = 1$}   
\psfrag{a4}[l][][0.6]{ $\alpha = 5$}  
\psfrag{b1}[l][][0.6]{ $\alpha = 0.05$}   
\psfrag{b2}[l][][0.6]{ $\alpha = 5$}      
\psfrag{b3}[l][][0.6]{ $\alpha = 50$}     
\psfrag{b4}[l][][0.6]{ $\alpha = 2.5e+04$}
    \includegraphics[width=0.39\textwidth]{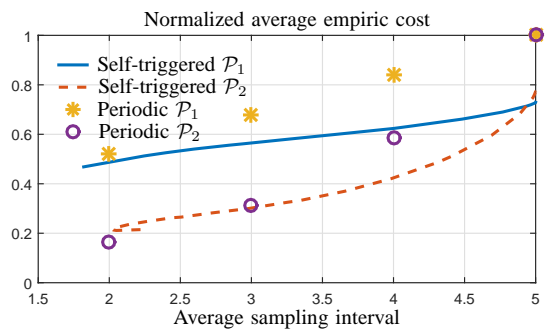}
    \caption{Performance for both processes for the self-triggered algorithm compared to a simple periodic algorithm for different sampling intervals. The max sampling interval $p^* = 5$. This is done by varying $\alpha$ from 0 to $5\cdot 10^4$ and calculating the average amount of samples for every value of $\alpha$.}
    \label{fig:self-vs-periodic-2N5}
\end{figure}

\begin{figure}[tp]
    \centering
    \psfrag{S1}[l][][0.8]{Self-triggered $\mathcal{P}_1$}
    \psfrag{S2}[l][][0.8]{Self-triggered $\mathcal{P}_2$}
    \psfrag{P1}[l][][0.8]{Periodic $\mathcal{P}_1$}
    \psfrag{P2}[l][][0.8]{Periodic $\mathcal{P}_2$}    
    \psfrag{AS}[c][][0.8]{Average sampling interval}
    \psfrag{IC}[c][][0.8]{Normalized average empiric cost}
\psfrag{a1}[l][][0.6]{ $\alpha = 0.005$}  
\psfrag{a2}[l][][0.6]{ $\alpha = 0.13$}   
\psfrag{a3}[l][][0.6]{ $\alpha = 1$}      
\psfrag{a4}[l][][0.6]{ $\alpha = 2.5$}    
\psfrag{a5}[l][][0.6]{ $\alpha = 10$}     
\psfrag{a6}[l][][0.6]{ $\alpha = 25$}     
\psfrag{a7}[l][][0.6]{ $\alpha = 2.5e+02$}
\psfrag{b1}[l][][0.6]{ $\alpha = 0.05$}   
\psfrag{b2}[l][][0.6]{ $\alpha = 1$}      
\psfrag{b3}[l][][0.6]{ $\alpha = 50$}     
\psfrag{b4}[l][][0.6]{ $\alpha = 5e+02$}  
\psfrag{b5}[l][][0.6]{ $\alpha = 2.5e+03$}
\psfrag{b6}[l][][0.6]{ $\alpha = 1.3e+04$}
\psfrag{b7}[l][][0.6]{ $\alpha = 5e+04$}  
    \includegraphics[width=0.39\textwidth]{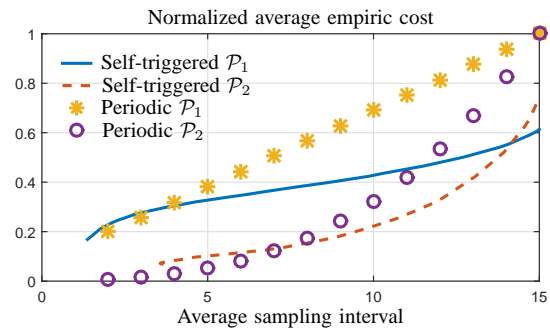}
    \caption{Performance for both processes for the self-triggered algorithm compared to a simple periodic algorithm for different sampling intervals. The max sampling interval $p^* = 15$. This is done by varying $\alpha$ from 0 to $10^9$ and calculating the average amount of samples for every value of $\alpha$.}
    \label{fig:self-vs-periodic-2N15}
\end{figure}

Both figures show that the self-triggered algorithm in general outperforms periodic sampling. The performance margin increases as the sampling interval increases. 
Fig.~\ref{fig:self-vs-periodic-2N5} shows that $\mathcal{P}_1$ when $\alpha = 0$ samples every 1.8 time steps on average, whereas $\mathcal{P}_2$ only samples every 2.2 time steps. When $\alpha$ increases both processes sample almost every 2 time steps. The reason for this is that the cost of the state for the down sampled systems in some cases is lower. 

The worse performance of the self-triggered algorithm at lower sampling intervals is more significant in Fig.~\ref{fig:self-vs-periodic-2N15} where the performance of process $\mathcal{P}_1$ shows a similar performance as in Fig.~\ref{fig:self-vs-periodic-single}. The lowest average sampling interval for process $\mathcal{P}_2$ is 3.6 time steps when $\alpha = 0$.  The self-triggered algorithm though significantly outperforms the periodic algorithm when the average sampling interval increases.

As $\alpha \rightarrow \infty$ the cost of sampling forces the self-triggered algorithm to behave similar to the periodic sampled algorithm. Therefore the performance gap between the periodic and self-triggered algorithms narrows, as the average sampling interval is close to $p^*$. 

\section{Conclusions}\label{sec:conclusions}
We have studied joint design of control and adaptive scheduling of multiple
loops, and have presented a method which at every sampling instant computes the
optimal control signal to be applied as well as the optimal time to wait before
taking the next sample. It is shown that this control law may be realized using
MPC and computed explicitly. The controller is also shown to be stabilizing
under mild assumptions. Simulation results show that the use of the presented
control law \edit{in most cases} may help reducing the required amount of communication without
almost any loss of performance compared to fast periodic sampling.

In the multiple loop case we have also presented an algorithm for guaranteeing
conflict free transmissions. It is shown that under mild assumptions there
always exists a feasible schedule for the network. \edit{The complexity of the multiple loop self-triggered MPC and the corresponding scheduling algorithm scales linearly in the number of loops.}

An interesting topic for future research is to further investigate the
complexity and possible performance increase for such an extended formulation.

\edit{Intuitively, additional performance gains can be achieved when the cost function takes process noise into account as well. Exploiting this could be of interest in future research.}

\edit{Another topic for future research would be to apply the presented framework to constrained control problems.}

\bibliographystyle{IEEEtran}
\bibliography{TCST2013}

\appendix

\subsection{Proof of Lemma~\ref{lem:P_conversion}}
\label{app:P_conversion}
\begin{proof}
Following Lemma~\ref{lem:sum_collapse} the problem is equivalent to
\begin{equation*}
\begin{aligned}
\min_{\mathcal{U}(i)} \sum_{r=0}^\infty\bigg(& x(k+i+r\cdot \rate)^TQ^\bir x(k+i+r\cdot \rate)\\
 +& u(k+i+r\cdot \rate)^TR^\bir u(k+i+r\cdot \rate)\\
 +& 2 x(k+i+r\cdot \rate)^TN^\bir u(k+i+r\cdot \rate) \bigg)\\
\end{aligned}
\end{equation*}
with
\begin{equation*}
x(k+i+(r+1)\cdot \rate) = A^\bir x(k+i+r\cdot \rate) + B^\bir u(k+i+r\cdot \rate).
\end{equation*}
This problem has the known optimal solution, see e.g. \cite{bert95},
$\|x(k+i)\|^2_{P^\bir}$. Where $P^\bir$ is given by the Riccati equation
(\ref{eq:ricc_periodic}), which has a solution provided that $0<R^\bir$,
implied by $0<R$, and $0 < Q^\bir$, implied by $0 < Q$, and that the pair
$(A^\bir,B^\bir)$ is controllable. Exactly what is stated in the Lemma.
\end{proof}

\subsection{Proof of Theorem~\ref{thm:min_J}}
\label{app:min_J}
\begin{proof}
From the Theorem we have that $0<Q$, $0<R$ and that the pair $(A^\bir,B^\bir)$
is controllable. Thus we may use Lemma~\ref{lem:P_conversion} to express the
cost (\ref{eq:J_w_constraints}) as
\begin{multline*}
    J(x(k),i,\mathcal{U}(i)) =\frac{\alpha}{i}+ \| x(k+i)\|^2_{P^\bir} \\
    + \sum_{l=0}^{i-1}\bigg(\| x(k+l)\|^2_Q + \|u(k)\|^2_R \bigg).
\end{multline*}
Now applying Lemma~\ref{lem:sum_collapse} we get
\begin{multline*}
    J(x(k),i,\mathcal{U}(i)) =\frac{\alpha}{i}+ \|x(k+i)\|^2_{P^\bir} \\
    + x(k)^TQ^\bi x(k)+ u(k)^TR^\bi u(k)+ 2x(k)^TN^\bi u(k).
\end{multline*}
with $x(k+i) =  A^\bi x(k)  +  B^\bi u(k)$. Minimizing
$J(x(k),i,\mathcal{U}(i))$ now becomes a finite horizon optimal control problem
with one prediction step into the future. This problem has the well defined
solution (\ref{eq:opt_open_loop_cost}), see \eg, \cite{bert95}, given by
iterating the Riccati equation (\ref{eq:ricc_i_step}).
\end{proof}

\subsection{Proof of Theorem~\ref{thm:stabilitybound}}\label{app:stabilitybound}
\begin{proof}
By assumption $(A,B)$ is controllable. Together with the choice of $\optrate$
this, via Lemma~\ref{lem:downsample}, implies that $(A^\biro,B^\biro)$ is
controllable. Further let $\hat x(k'|k)$ denote an estimate of $x(k')$, given
all available measurements up until time $k$. Defining
\begin{equation}\label{eq:S_i_def}
 \|\hat x(k|k)\|^2_{S^\bi} \triangleq \sum_{l=0}^{i-1}\bigg(\|\hat x(k+l|k)\|^2_Q + \|\hat u(k|k)\|^2_R \bigg)
\end{equation}
we may, since by assumption $0<Q$ and $0<R$, use Lemma~\ref{lem:P_conversion}
and Theorem~\ref{thm:min_J} to express $V_k$, the optimal value of the cost
(\ref{eq:J_w_constraints}) at the current sampling instant $k$, as
\begin{equation}\label{eq:V_k}
\begin{aligned}
    V_k \triangleq & \min_{{i\in\nomI},\hat u(k|k)} J(x(k),i,\hat u(k|k))\\
    =&\min_{i\in\nomI}\frac{\alpha}{i}+ \|\hat x(k+i|k)\|^2_{P^\biro} + \|\hat x(k|k)\|^2_{S^\bi}\\
    =&\min_{i\in\nomI}\frac{\alpha}{i}+ \|\hat x(k|k)\|^2_{P^\bi}.
    \end{aligned}
\end{equation}
We will use $V_k$ as a Lyapunov-like function. Assume that $V_{k+i}$ is the
optimal cost at the next sampling instant $k+i$. Again using
Theorem~\ref{thm:min_J} we may express it as
\begin{equation}\label{eq:V_k_i}
\begin{aligned}
    V_{k+i} \triangleq & \min_{{j\in\nomI},\hat u(k+i|k+i)} J(x(k+i),j,\hat u(k+i|k+i))\\
    \leq&\min_{\hat u(k+i|k+i)} J(x(k+i),j=\optrate,\hat u(k+i|k+i))\\
    =&\frac{\alpha}{\optrate}+ \|\hat x(k+i|k+i)\|^2_{P^\biro} \\
    =& \frac{\alpha}{\optrate}+ \|\hat x(k+i|k)\|^2_{P^\biro}.
    \end{aligned}
\end{equation}
Where the inequality comes from the fact that choosing $j=\optrate$ is
sub-optimal. Taking the difference we get
\begin{equation*}
V_{k+i} - V_k \leq \frac{\alpha}{\optrate}-\frac{\alpha}{i} - \|\hat x(k|k)\|^2_{S^\bi}
\end{equation*}
which in general is not decreasing. However we may use the following idea to
bound this difference: Assume that there $\exists\, \epsilon\in(0,1]$ and
$\beta \in\mathbb{R^+}$ such that we may write
\begin{equation*}
V_{k+i} - V_k \leq -\epsilon V_k + \beta,
\end{equation*}
\edit{for all $k,k+i \in \mathcal{D}$, see (\ref{eq:V_k}) and (\ref{eq:set_of_samples}).}
Thus, at $l$ sampling instances into the
future, which happens at let's say time $k+l'$, we have that
\begin{equation*}
V_{k+l'} \leq (1-\epsilon)^l \cdot V_k + \beta \cdot \sum_{r=0}^{l-i}(1-\epsilon)^l.
\end{equation*}
Since $\epsilon \in (0,1]$ this is equivalent to
\begin{equation*}
V_{k+l'} \leq (1-\epsilon)^l \cdot V_k + \beta \cdot \frac{1-(1-\epsilon)^l}{1-(1-\epsilon)},
\end{equation*}
which as $l\rightarrow\infty$ gives us an upper bound on the cost function,
$V_{k+l'} \leq {\beta}/{\epsilon}$. Applying this idea on our setup we should
fulfil
\begin{multline*}
\frac{\alpha}{\optrate}-\frac{\alpha}{i} - \|\hat x(k|k)\|_{S^\bi} \\
\leq-\epsilon\frac{\alpha}{i}- \epsilon\|\hat x(k+i|k)\|^2_{P^\biro} - \epsilon
\|\hat x(k|k)\|^2_{S^\bi} + \beta.
\end{multline*}
Choosing $\beta \triangleq {\alpha}/{\optrate}-(1-\epsilon){\alpha}/{\maxI}$ we have
fulfilment if
\begin{equation}\label{eq:epsilon_S}
\begin{aligned}
\epsilon\Big(\|\hat x(k+i|k)\|^2_{P^\biro} + \|\hat x(k|k)\|^2_{S^\bi}\Big)   \leq \|\hat x(k|k)\|^2_{S^\bi}.
    \end{aligned}
\end{equation}
Clearly there $\exists \epsilon \in(0,1]$ such that the above relation is
fulfilled if $0<\|\hat x(k|k)\|^2_{S^\bi}$. For the case $\|\hat x(k|k)\|^2_{S^\bi}=0$ we
must, following the definition (\ref{eq:S_i_def}) and the assumption $0<Q$,
have that $\hat x(k|k)=0$ and $\hat x(k+i|k)=0$ and hence the relation is
fulfilled also in this case. Using the final step in (\ref{eq:V_k}) we may
express (\ref{eq:epsilon_S}) in easily computable quantities giving the
condition
\begin{equation*}
\begin{aligned}
\|\hat x(k+i|k)\|^2_{P^\biro} \leq (1-\epsilon)\|\hat x(k|k)\|^2_{P^\bi}.
    \end{aligned}
\end{equation*}
As this should hold $\forall x$ and $\forall i\in\nomI$ we must fulfill
\begin{equation*}
 \big( A^\bi- B^\bi L^\bi\big)^TP^\biro \big( A^\bi- B^\bi L^\bi\big) \leq (1-\epsilon)P^\bi
\end{equation*}
which is stated in the Theorem. Summing up we have
\begin{equation*}
V_{k+l'} \leq \frac{\alpha}{\epsilon}\bigg(\frac{1}{\optrate}-(1-\epsilon)\frac{1}{\maxI}\bigg)
\end{equation*}
which  is minimized by maximizing $\epsilon$. From the definition of the cost
(\ref{eq:J}) we may also conclude that ${\alpha}/{\maxI} \leq V_{k+l'}$. With
\begin{equation*}
V_{k+l'} = \min_{i\in\nomI}\bigg(\frac{\alpha}{i}+ \|\hat x(k+l'|k+l')\|^2_{P^\bi}\bigg)
\end{equation*}
we may conclude that
\begin{equation*}
\frac{\alpha}{\maxI}\leq \lim_{k\rightarrow\infty} \min_{i\in\nomI}\bigg(\frac{\alpha}{i}+ \|\hat x(k|k)\|^2_{P^\bi}\bigg) \leq \frac{\alpha}{\epsilon}\bigg(\frac{1}{\optrate}-(1-\epsilon)\frac{1}{\maxI}\bigg).
\end{equation*}
\end{proof}

\subsection{Proof of Lemma~\ref{lem:pstarIsGamma}}\label{app:pstarIsGamma}
\begin{proof}
From (\ref{eq:stab_rate_cond}) it is clear that $\optrate=\maxI$ if
$\nexists\lambda\in\lambda(A)$ except $\lambda=1$ such that $\lambda^\maxI=1$.
In polar coordinates we have that
$\lambda=|\lambda|\exp({j\cdot\angle\lambda})$ implying
$\lambda^\maxI=|\lambda|^\maxI \exp({j\cdot\maxI\cdot\angle\lambda})=1$ may
only be fulfilled if $|\lambda|=1$ and $\angle\lambda=\frac{2\pi}{\maxI}\cdot
n$ for some $n\in\mathbb{N}^+$, which contradicts
Assumption~\ref{ass:pstarIsGamma}.
\end{proof}

\subsection{Proof of Corollary~\ref{corr:stability}}\label{app:stabilitycorr}
\begin{proof}
From Theorem~\ref{thm:stabilitybound} we have that as $k\rightarrow\infty$
\begin{equation*}
\frac{\alpha}{\maxI}\leq \min_{i\in\nomI}\bigg(\frac{\alpha}{i}+ \|x(k)\|^2_{P^\bi}\bigg) \leq \frac{\alpha}{\epsilon}\bigg(\frac{1}{\optrate}-(1-\epsilon)\frac{1}{\maxI}\bigg)
\end{equation*}
which as $\optrate=\maxI$, given by Lemma~\ref{lem:pstarIsGamma}, simplifies to
\begin{equation*}
\frac{\alpha}{\maxI}\leq \min_{i\in\nomI}\bigg(\frac{\alpha}{i}+ \|x(k)\|^2_{P^\bi}\bigg) \leq \frac{\alpha}{\maxI}
\end{equation*}
independent of $\epsilon$. Implying that as $k\rightarrow\infty$ we have that
$i=\gamma$ and $\|x(k)\|^2_{P^\bi} = 0$, since $0<P^\bi$ provided $0<Q$ this
implies $x(k)=0$ independent of $\alpha$. In the case $\alpha=0$ the bound from
Theorem~\ref{thm:stabilitybound} simplifies to that as $k\rightarrow\infty$
\begin{equation*}
0\leq \min_{i\in\nomI}\|x(k)\|^2_{P^\bi} \leq 0
\end{equation*}
and hence for the optimal $i$ we have $\|x(k)\|^2_{P^\bi}=0$ implying $x(k)=0$
as above.
\end{proof}

\subsection{Proof of Lemma~\ref{lem:persistent}}\label{app:persistent}
\begin{proof}
When loop~$q$ was last sampled at time $k_q<k_\ell$ it was optimized over
$\mathcal{I}_q(k_q)$ and found the optimal feasible time until the next sample
$I_q(k_q)$. The loop then reserved the infinite sequence
\begin{multline*}
S_q(k_q)=\{k_q+I_q(k_q),\ k_q+I_q(k)+\rate_q,\\
 k_q+I_q(k_q)+2\rate_q,\ k_q+I_q(k_q)+3\rate_q,\ \ldots\}.
\end{multline*}
When loop~$\ell$ now should choose $I_\ell(k_\ell)$ it must be able to reserve
(\ref{eq:samplingpattern}). To ensure that this is true it must choose
$I_\ell(k_\ell)$ so that $S_\ell(k_\ell)\cap S_q(k_q)=\emptyset$,
$\forall\,q\in\loopset\setminus\{\ell\}$. This holds if we have that
\begin{equation*}
\begin{aligned}
&k_\ell+I_\ell(k_\ell)+m\cdot\rate_\ell \neq k_q+I_q(k_q)+n\cdot\rate_q, \\
  &\forall\,m,n\in\mathbb{N},\, \forall\,q\in\loopset\setminus\{\ell\}
\end{aligned}
\end{equation*}
Simplifying the above condition we get conditions on the feasible values of
$I_\ell(k_\ell)$
\begin{equation*}
\begin{aligned}
I_\ell(k_\ell)&\neq \Big(k_q+I_q(k_q)\Big)-k_\ell+n\cdot\rate_q-m\cdot\rate_\ell, \\
 &\forall\,m,n\in\mathbb{N},\, \forall\,q\in\loopset\setminus\{\ell\}
\end{aligned}
\end{equation*}
Noticing that $\Big(k_q+I_q(k_q)\Big)$ is the next transmission of loop~$q$ we
denote it $k_q^\text{next}$, giving the statement in the lemma.
\end{proof}

\subsection{Proof of Theorem~\ref{thm:alwayspersistent}}\label{app:alwayspersistent}
\begin{proof}
Assuming $\rate_\ell=\rate$ and $\nomI_\ell=\nomI$ for all loops
Lemma~\ref{lem:persistent} gives $\mathcal{I}_\ell(k_\ell)$ as stated in the
theorem. Since by assumption $\loopset\subseteq\nomI$ we know that
\begin{equation*}
\begin{aligned}
\mathcal{I}_\ell(k_\ell)&\supseteq\{i\in\loopset|i\neq k_q^\text{next}-k_\ell+r\cdot \rate,\,r\in\mathbb{Z},\, q\in\loopset\setminus\{\ell\} \}.
\end{aligned}
\end{equation*}
Further, since $\max\loopset\leq\rate$ we know that the set
\begin{equation*}
\begin{aligned}
\{i\in\loopset|i = k_q^\text{next}-k_\ell+r\cdot \rate,\,r\in\mathbb{Z}\}
\end{aligned}
\end{equation*}
contains at most one element for a given loop $q$. Thus,
\begin{equation*}
\begin{aligned}
\{i\in\loopset|i\neq k_q^\text{next}-k_\ell+r\cdot \rate,\,r\in\mathbb{Z},\, q\in\loopset\setminus\{\ell\} \}
\end{aligned}
\end{equation*}
contains at least one element as
$q\in\loopset\setminus\{\ell\}\subset\loopset$. Hence
$\mathcal{I}_\ell(k_\ell)$ always contains at least one element
$\forall\,k_\ell$, and thus there always exists a feasible time to wait.
\end{proof}

\subsection{Proof of Theorem~\ref{thm:multistabilitybound}}\label{app:multistabilitybound}
\begin{proof}
Direct application of Theorem~\ref{thm:stabilitybound} on each loop. The
corresponding proof carries through as the choice of $\mathcal{I}_\ell(k)$
together with the remaining assumptions guarantees that we may apply
Theorem~\ref{thm:min_J_multi} on the feasible values of $i$ in every step, thus
the expression for $V_k$ in (\ref{eq:V_k}) always exists. The critical step is
that the upper bound on $V_{k+i}$ in (\ref{eq:V_k_i}) must exist, \ie, the
choice $j=\optrate$ must be feasible. This is also guaranteed by the choice of
$\mathcal{I}_\ell(k)$, see Remark~\ref{rem:pstarInNext}.
\end{proof}

\end{document}